\documentclass[12pt]{article}
\usepackage{amsfonts}
\usepackage{color}
\newcommand{\Section}[1]%
{\section{#1}\setcounter{equation}{0}%
\setcounter{theorem}{0}}
\newtheorem{theorem}{Theorem}
\newtheorem{lemma}[theorem]{Lemma}
\newtheorem{coro}[theorem]{Corollary}
\def\re{\mathbb{R}}
\def\co{\mathbb{C}}

\def\ze{\mathbb{Z}}

\newenvironment{proof}[1]%
{\par\noindent{\em #1:\ }}%
{~\rule{2mm}{2mm}\par\bigskip}
\begin{document}
%%%%%%%%%%%%%%%%%%%%%%%%%%%%%%%%%%%%%%%%%%%
%%%%%%%%%%%%%%%%%%%%%%%%%%%%%%%%%%%%%%%%%%%
%%%%%%%%%%%%%%%%%%%%%%%%%%%%%%%%%%%%%%%%%%%
\newpage\thispagestyle{empty}
{\topskip 2cm
\begin{center}
{\Large\bf Maximum Spontaneous Magnetization and Nambu-Goldstone Mode\\} 
\bigskip\bigskip
{\Large Tohru Koma\footnote{\small \it Department of Physics, Gakushuin University, Mejiro, Toshima-ku, Tokyo 171-8588, JAPAN,
{\small\tt e-mail: tohru.koma@gakushuin.ac.jp}}
\\}
\end{center}
\vfil
\noindent
{\bf Abstract:} We study quantum antiferromagnetic Heisenberg models on a hypercubic lattice. 
We prove the following three theorems without any assumption: 
(i) The spontaneous magnetization which is obtained by applying the infinitesimally weak 
symmetry breaking field is equal to the maximum spontaneous magnetization at zero or non-zero low temperatures. 
(ii) When the spontaneous magnetization is non-vanishing at zero temperature, 
there appears a gapless excitation, Nambu-Goldstone mode, above an infinite-volume pure ground state. 
(iii) When the spontaneous magnetization is non-vanishing at zero or non-zero low temperatures,  
the transverse correlation in the infinite-volume limit exhibits a Nambu-Goldstone-type slow decay.
\par
%%%%%%%%%%%%%%%%%%%%%%%%%%%%%%%%%%%%%%%%%%%%%%%%%%%%%%%%%%%%%%%%%%%%%%%%%%%
\noindent
\bigskip
\hrule
\bigskip
%%%%%%%%%%%%%%%%%%%%%%%%%%%%%%%%%%%%%%%%%%%%%%%%%%%%%%
\vfil}
%\newpage
%%%%%%%%%%%%%%%%%%%%%%%%%%%%%%%%%%%%%%%%%%%

%%%%%%%%%%%%%%%%%%%%%%%%%%%%%%%%%%%%%%%%%%%%%%%%%%%%%%
\Section{Introduction}
\label{Intro}

For quantum many-body systems, the excitation spectrum for low energy states above the ground state 
has been often computed by using trial wavefunctions within Bijl-Feynman single-mode approximation \cite{Bijl,Feynman,{Stringari}}. 
By relying on the method, Momoi \cite{Momoi} obtained a spin-wave spectrum above a symmetry-breaking ground state 
with a N\'eel order in Heisenberg antiferromagnets. 
He also evaluated \cite{Momoi2} the decay of the transverse spin-spin correlation which is related to 
Nambu-Goldstone mode \cite{Nambu,NJL,Goldstone,GSW}. 
His results agree with the expected ones from Nambu-Goldstone argument for continuous symmetry 
breaking. However, it is well known \cite{KomaTasaki2,Tasaki} that there appear many low-lying eigenstates whose excitation energy is 
very close to the energy of the symmetric ground state of the finite-volume Hamiltonian, and 
that these low-lying eigenstates yield a set of symmetry-breaking ground states in the infinite-volume limit 
by forming linear combinations of the low-lying eigenstates and the symmetric ground state. 
Therefore, in order to obtain the true spectrum of low-energy excitations above an infinite-volume pure ground state, 
we have to distinguish them from the low-lying eigenstates which yield a set of infinite-volume ground states. 

In this paper, we improve Momoi's argument. In consequence, we prove the existence of the Nambu-Goldstone mode 
above an infinite-volume pure ground state in quantum antiferromagnetic Heisenberg models on a hypercubic lattice. 
We also prove that the transverse spin-spin correlation, which is related to 
Nambu-Goldstone mode, exhibits a certain slow decay when the spontaneous magnetization 
exhibits the non-vanishing maximum value at zero or non-zero low temperatures.   
 
In the next section, we present the precise definition of 
the Hamiltonian of the quantum antiferromagnetic Heisenberg models, and describe our main theorems. 
The rest of Sections are devoted to the proofs of the main theorems as follows:  
The maximum spontaneous magnetization at zero and non-zero temperatures is treated 
in Sections~\ref{sec:SMMaxZero} and \ref{Sec:SponMagFiniteT}, respectively. 
The existence of the Nambu-Goldstone mode is proved in Sec.~\ref{sec:NambuGoldstone}. 
An alternative proof of the existence of the mode is given in Sec.~\ref{sec:AltNambuGoldstone}.  
In Sec.~\ref{sec:TransCorr}, we prove that 
the transverse spin-spin correlation exhibits a Nambu-Goldstone-type slow decay. 
Appendices~\ref{App::BogolyIneq}-\ref{appen:chibound} are devoted to technical estimates.  

%%%%%%%%%%%%%%%%%%%%%%%%%%%%%%%%%%%%%%%%%%%
\Section{Models and Main Results}

\subsection{Quantum Heisenberg antiferromagnets}

In the present paper, we consider quantum Heisenberg antiferromagnets 
which have reflection positivity \cite{DLS,KLS}. 
The extension of our method to anisotropic antiferromagnets is relatively straightforward \cite{KLS2,KuboKishi}. 
More precisely, we can treat the Hamiltonian $H_{0,{\rm p}}^{(\Lambda)}$ of (\ref{H0}) below with an additional Ising term. 

Let $\Gamma$ be a finite subset of the $d$-dimensional hypercubic lattice $\ze^d$, i.e., $\Gamma\subset\mathbb{Z}^d$,  
with $d\ge 1$. For each site $x=(x^{(1)},x^{(2)},\ldots,x^{(d)})\in\Gamma$, 
we associate three component quantum spin operator ${\bf S}_x=(S_x^{(1)},S_x^{(2)},S_x^{(3)})$ 
with magnitude of spin, $S=1/2,1,3/2,2,\ldots$. More precisely, the spin operators, $S_x^{(1)}, S_x^{(2)}, S_x^{(3)}$, 
are $(2S+1)\times(2S+1)$ matrices at the site $x$. They satisfy the commutation relations, 
$$
[S_x^{(1)},S_x^{(2)}]=iS_x^{(3)}, \quad [S_x^{(2)},S_x^{(3)}]=iS_x^{(1)}, 
\quad \mbox{and} \quad [S_x^{(3)},S_x^{(1)}]=iS_x^{(2)},
$$
and $(S_x^{(1)})^2+(S_x^{(2)})^2+(S_x^{(3)})^2=S(S+1)$ for $x\in\Gamma$. 
For the finite lattice $\Gamma$, the whole Hilbert space is given by 
$$
\mathfrak{H}_\Gamma=\bigotimes_{x\in\Gamma} \co^{2S+1}.
$$
More generally, the algebra of observables on $\mathfrak{H}_\Gamma$ is given by 
$$
\mathfrak{A}_\Gamma:=\bigotimes_{x\in\Gamma}M_{2S+1}(\co),
$$
where $M_{2S+1}(\co)$ is the algebra of $(2S+1)\times(2S+1)$ complex matrices. 
When two finite lattices, $\Gamma_1$ and $\Gamma_2$, satisfy $\Gamma_1\subset\Gamma_2$, 
the algebra $\mathfrak{A}_{\Gamma_1}$ is embedded in $\mathfrak{A}_{\Gamma_2}$ by 
the tensor product $\mathfrak{A}_{\Gamma_1}\otimes I_{\Gamma_2\backslash\Gamma_1}\subset 
\mathfrak{A}_{\Gamma_2}$ with the identity $I_{\Gamma_2\backslash\Gamma_1}$. 
The local algebra is given by 
$$
\mathfrak{A}_{\rm loc}=\bigcup_{\Gamma\subset\ze^d:|\Gamma|<\infty}\mathfrak{A}_{\Gamma},
$$
where $|\Gamma|$ is the number of the sites in the finite lattice $\Gamma$. 
The quasi-local algebra is defined by the completion of the local algebra $\mathfrak{A}_{\rm loc}$ 
in the sense of the operator-norm topology. 

Consider a $d$-dimensional finite hypercubic lattice, 
\begin{equation}
\label{Lambda}
\Lambda:=\{-L+1,-L+2,\ldots,-1,0,1,\ldots,L-1,L\}^d\subset\mathbb{Z}^d,
\end{equation}
with a large positive integer $L$ and $d\ge 1$. 
The Hamiltonian $H_{\rm p}^{(\Lambda)}(B)$ of the Heisenberg antiferromagnet on the lattice $\Lambda$ is given by 
\begin{equation}
\label{H(B)}
H_{\rm p}^{(\Lambda)}(B)=H_{0,{\rm p}}^{(\Lambda)}-BO^{(\Lambda)},
\end{equation}
where the first term in the right-hand side is the Hamiltonian of the nearest neighbor spin-spin 
antiferromagnetic interactions, 
\begin{equation}
\label{H0}
H_{0,{\rm p}}^{(\Lambda)}:=\sum_{\{x,y\}\subset\Lambda:\;|x-y|=1}{\bf S}_x\cdot{\bf S}_y,
\end{equation}
and the second term is the potential due to the external magnetic field $B\in\re$ with the order parameter, 
$$
O^{(\Lambda)}:=\sum_{x\in\Lambda}(-1)^{x^{(1)}+x^{(2)}+\cdots+x^{(d)}}S_x^{(3)}.
$$
Here, the subscript ${\rm p}$ of the Hamiltonian $H_{0,{\rm p}}^{(\Lambda)}$ denotes the periodic 
boundary condition.  

%%%%%%%%%%%%%%%%%%%%%%%%%%%%%%%%%%%%%%%
\subsection{Zero temperature}
\label{ZeroTempThms}

We first describe our main results for the ground states. 
Let $\Phi_0^{(\Lambda)}(B)$ be a ground-state vector 
of the Hamiltonian $H_{\rm p}^{(\Lambda)}(B)$. We can take the vector $\Phi_0^{(\Lambda)}(B)$ 
to be translationally invariant with period 2 because of the periodic boundary condition. 
The infinite-volume ground state is given by 
\begin{equation}
\label{omegaPhi0B}
\omega_{\Phi_0(B)}(\cdots):={\rm weak}^\ast\mbox{-}\lim_{\Lambda\nearrow\ze^d}
\langle \Phi_0^{(\Lambda)}(B),(\cdots)\Phi_0^{(\Lambda)}(B)\rangle. 
\end{equation}
Here, we can take a suitable sequence of finite lattices $\Lambda$ going to $\ze^d$ 
so that the expectation value converges to a linear functional for the set of 
the quasi-local algebra. Actually, it is well know that, 
for any given uniformly bounded sequence $\{\omega_n\}_n$ of linear functionals on a separable normed linear space, 
there exists a subsequence $\{\omega_{n_k}\}_k$ of the sequence $\{\omega_n\}_n$ such that 
$\omega_{n_k}$ converges to a linear functional $\omega_\infty$ on the space. 
(See, e.g., Sec.~IV.5 in the book~\cite{ReedSimon}.) Throughout the present paper, we use weak$^\ast$-limit 
in this sense. Of course, the resulting linear functional $\omega_\infty$ may be depend on the subsequence.
Therefore, in general, there may appear many infinite-volume ground states $\omega_{\Phi_0(B)}=\omega_{\Phi_0(B)}(\cdots)$, 
depending on the sequence of the finite lattices $\Lambda$.

Write 
\begin{equation}
\label{mPhi0B}
m[\Phi_0(B)]:=\frac{1}{|\Lambda|}\omega_{\Phi_0(B)}(O^{(\Lambda)}). 
\end{equation}
Clearly, this quantity does not depend on 
the even side length $2L$ of the hypercubic lattice $\Lambda$ of (\ref{Lambda}) 
because of the translational invariance, and one has  
\begin{equation}
\label{mPhi0Bexp}
m[\Phi_0(B)]=\lim_{\Lambda\nearrow\ze^d}\frac{1}{|\Lambda|}
\langle \Phi_0^{(\Lambda)}(B),O^{(\Lambda)}\Phi_0^{(\Lambda)}(B)\rangle,
\end{equation}
where we choose the sequence of the finite lattices $\Lambda$ to be the same as in (\ref{omegaPhi0B}). 
Therefore, the spontaneous magnetization is formally given by 
\begin{equation}
\label{ms0}
\lim_{B \searrow 0}m[\Phi_0(B)]=\lim_{B\searrow 0}
\lim_{\Lambda\nearrow\ze^d}\frac{1}{|\Lambda|}
\langle \Phi_0^{(\Lambda)}(B),O^{(\Lambda)}\Phi_0^{(\Lambda)}(B)\rangle
\end{equation}
for the ground state $\Phi_0^{(\Lambda)}(B)$ in the infinite-volume limit.

In order to make the definition of the spontaneous magnetization more precise, 
we recall the notion of infinite-volume ground states. 
Let $\omega$ be a positive linear functional for the quasi-local algebra. Then, if $\omega$ satisfies \cite{BR}
\begin{equation}
\label{defGS}
\lim_{\Lambda\nearrow \ze^d}\omega(a^\ast[H_{0,{\rm p}}^{(\Lambda)},a])\ge 0
\end{equation}
for any $a\in\mathfrak{A}_{\rm loc}$, we say that $\omega$ is an infinite-volume ground state in the case of 
the external magnetic field $B=0$. The physical meaning of the inequality 
is that any local perturbation cost non-negative energy. 
Clearly, the above limit exists for a fixed $a\in\frak{A}_{\rm loc}$, and the condition does not depend on 
the boundary condition of the Hamiltonian. Therefore, instead of the periodic boundary condition, 
one can impose any boundary condition. Correspondingly, there may appear many types of infinite-volume ground states. 

Write 
\begin{equation}
\label{omegaPhi00}
\omega_{\Phi_0(+0)}:={\rm weak}^\ast\mbox{-}\lim_{B\searrow 0}\omega_{\Phi_0(B)},
\end{equation}
where we take a suitable sequence of $B$ going to zero  
so that the expectation value converges to a linear functional for the set of 
the quasi-local algebra, again. Then, one can easily check that the state $\omega_{\Phi_0(+0)}$ fulfills 
the condition (\ref{defGS}) for infinite-volume ground states, 
and hence it is the corresponding infinite-volume ground state.

For an infinite-volume ground state $\omega$, the supremum of the spontaneous magnetization is given by 
$$
\lim_{N\nearrow\infty}\sup_{|\Lambda|\ge N}\frac{1}{|\Lambda|}\omega(O^{(\Lambda)}).
$$
Here, $\Lambda$ is the hypercubic lattice of (\ref{Lambda}) with the even side length $2L$. 
In addition, if $\omega$ is translationally invariant with period 2, then we say that 
$\omega$ is a translationally invariant infinite-volume ground state. 
For the translationally invariant state $\omega$, the spontaneous magnetization is defined by 
$$
\lim_{\Lambda\nearrow\ze^d}\frac{1}{|\Lambda|}\omega(O^{(\Lambda)})
$$ 
with the hypercubic lattice $\Lambda$ of (\ref{Lambda}) with the even side length $2L$. 
Because of the translational invariance of $\omega$, the limit is equal to $\omega(O^{(\Lambda)})/|\Lambda|$.  
Similarly, by relying on the translational invariance 
of the state $\omega_{\Phi_0(+0)}$ of (\ref{omegaPhi00}), we define 
\begin{equation}
\label{ms}
m_{\rm s}:=\lim_{\Lambda\nearrow\ze^d}\frac{1}{|\Lambda|}\omega_{\Phi_0(+0)}(O^{(\Lambda)}).
\end{equation}
This is equal to $\omega_{\Phi_0(+0)}(O^{(\Lambda)})/|\Lambda|$, and is equal to the quantity of (\ref{ms0}) 
if the same sequence as that in (\ref{omegaPhi00}) is taken in the limit $B\searrow 0$.
This quantity $m_{\rm s}$ is nothing but the spontaneous magnetization for the infinite-volume ground state 
$\omega_{\Phi_0(+0)}$. The spontaneous magnetization $m_{\rm s}$ may depend on the two sequences in 
the two weak$^\ast$-limits, $\Lambda\nearrow\ze^d$ and $B\searrow 0$, because the ground state $\omega_{\Phi_0(+0)}$ 
is constructed by these sequences.

\begin{theorem}
\label{thm:sponmag}
The value $m_{\rm s}$ of (\ref{ms}) is independent of the choices of the sequences in the two 
weak$^\ast$-limits, $\Lambda\nearrow\ze^d$ and $B\searrow 0$, in  (\ref{omegaPhi0B}) and (\ref{omegaPhi00}), 
and hence the value $m_{\rm s}$ of the spontaneous magnetization is uniquely determined. 
Further, $m_{\rm s}$ is equal to the supremum of the spontaneous magnetizations  
over all of the infinite-volume ground states.    
\end{theorem}

Since the supremum $m_{\rm s}$ of the spontaneous magnetization is realized 
for the infinite-volume ground state $\omega_{\Phi_0(+0)}$ of (\ref{omegaPhi00}), 
we call $m_{\rm s}$ the maximum spontaneous magnetization. 
Of course, one can expect that $m_{\rm s}$ takes the maximum value over all the infinite-volume 
ground states because the weak external magnetic field in (\ref{omegaPhi00}) forces the whole system to be magnetized 
in the favorable direction. However, as far as we know, this statement had not yet been proved true.
The proof of Theorem~\ref{thm:sponmag} is given in Sec.~\ref{sec:SMMaxZero}. 

We write 
\begin{equation}
\label{Grstate}
\omega_B^{(\Lambda)}(\cdots):=\frac{1}{q^{(\Lambda)}(B)}{\rm Tr}\;[P_0^{(\Lambda)}(B)(\cdots)]
\end{equation}
for the expectation value for the ground state of the Hamiltonian $H_{\rm p}^{(\Lambda)}(B)$, 
where $P_0^{(\Lambda)}(B)$ is the projection onto the sector of the ground states and 
$q^{(\Lambda)}(B)$ is the degeneracy of the ground states, i.e., $q^{(\Lambda)}(B)={\rm Tr}\;P_0^{(\Lambda)}(B)$. 
We also write 
\begin{equation}
\label{omega0}
\omega_0(\cdots):={\rm weak}^\ast\mbox{-}\lim_{B\searrow 0}{\rm weak}^\ast\mbox{-}\lim_{\Lambda \nearrow\ze^d}
\omega_B^{(\Lambda)}(\cdots)
\end{equation}
for the infinite-volume ground state with zero magnetic field $B=0$. 
Here, if necessary, we take suitable sequences with respect to $\Lambda$ 
and $B$ in the two weak$^\ast$ limits so that the expectation value converges to a linear functional in the same way 
as in (\ref{omegaPhi0B}). From Theorem~\ref{thm:sponmag}, 
this state $\omega_0$ also exhibits the maximum spontaneous magnetization $m_{\rm s}$. Therefore, if an infinite-volume 
ground state exhibits a non-vanishing spontaneous magnetization, then the state $\omega_0$ shows 
the non-vanishing maximum spontaneous magnetization $m_{\rm s}$. 

Next, in order to describe our theorem about the Nambu-Goldstone mode, we recall the notion of 
pure states \cite{BR1}: 
A state is called pure whenever it cannot be expressed as a convex combination of other states. 
In order to prove the existence of the corresponding gapless mode above the pure states, we rely on 
the same idea as in the condition (\ref{defGS}) for the infinite-volume ground states. 
Namely, we construct a low energy excitation above a ground state by using a local perturbation $a$. 
However, the state which is obtained by simply exciting the ground state by a local operator may 
contain a component in the sector of the ground states. In order to eliminate the contribution of 
the ground states, we introduce a spectral cutoff function $\hat{g}$ as follows:   
Consider the time evolution \cite{BR} of local 
operator $a$, 
$$ 
\tau_{t,B}^{(\Lambda)}(a):=\exp[iH_{\rm p}^{(\Lambda)}(B)t]a\exp[-iH_{\rm p}^{(\Lambda)}(B)t], 
$$
and 
$$
\tau_{\ast g,B}^{(\Lambda)}(a):=\int_{-\infty}^{+\infty}dt\; g(t)\tau_{t,B}^{(\Lambda)}(a),
$$
where the function $g$ is the Fourier transform of the function $\hat{g}\in C_0^\infty(\re)$. 
We choose the function $\hat{g}$ 
to satisfy ${\rm supp}\;\hat{g}\subseteq(0,\gamma)$ with a positive constant $\gamma$.
Write \cite{BR}
$$
\tau_{\ast g,0}(a):=\lim_{B\searrow 0}\lim_{\Lambda\nearrow\ze^d}\tau_{\ast g,B}^{(\Lambda)}(a),
$$
where the sequences in the double limit are taken to be the same as those in the ground state $\omega_0$ of (\ref{omega0}). 
Consider 
\begin{equation}
\label{ExcitationEnergy}
\lim_{\Lambda\nearrow\ze^d}
\frac{\omega_0([\tau_{*g,0}(a)]^\ast[H_{\rm p}^{(\Lambda)}(0),\tau_{*g,0}(a)])}
{\omega_0([\tau_{*g,0}(a)]^\ast\tau_{*g,0}(a))}.
\end{equation}
Since the energy cutoff function $\hat{g}$ satisfies ${\rm supp}\;\hat{g}\subseteq(0,\gamma)$, 
this quantity gives a pure excitation energy above the ground state $\omega_0$.
Therefore, if a local operator $a$ can be chosen to satisfy that the quantity (\ref{ExcitationEnergy}) 
is smaller than any given small energy, then there exists a gapless excitation above a pure infinite-volume 
ground state.

The existence of the Nambu-Goldstone mode follows from the following theorem:  

\begin{theorem}
\label{thm:GSgapless}
Suppose that the spontaneous magnetization is non-vanishing 
for an infinite-volume ground state. 
Then, there exists an infinite-volume pure ground state such that 
there appears a gapless excitation above the ground state and that the state exhibits 
the non-vanishing maximum spontaneous magnetization $m_{\rm s}$.  
\end{theorem}
\medskip

The proof is given in Sec.~\ref{sec:NambuGoldstone}. 
An alternative proof of the existence of the gapless mode is given in Sec.~\ref{sec:AltNambuGoldstone}. 
%by relying on Lieb-Robinson bounds \cite{LR,NS,HK} 
\medskip 

\noindent
{\it Remark:} In \cite{Wreszinski,LFPW}, the following statement was proved for generic quantum spin systems: 
A non-vanishing spectral gap above a unique infinite-volume ground state in the domain of 
the infinite-volume Hamiltonian implies 
that there is no spontaneous symmetry breaking. However, the sector of the infinite-volume ground 
states may be degenerate. In fact, as mentioned in Introduction, 
there appear many low-lying eigenstates whose excitation energy is 
very close to the energy of the symmetric ground state of the finite-volume Hamiltonian, and 
these low-lying eigenstates yield a set of symmetry-breaking ground states in the infinite-volume limit 
by forming linear combinations of the low-lying eigenstates and the symmetric ground state \cite{KomaTasaki2,Tasaki}. 
It is not clear why such a contribution of many low-lying states yields a unique ground state. 
In \cite{Momoi}, the contribution was not taken into account, too. 
Wreszinski \cite{Wreszinski2} assumed an ergodic property with respect to 
the time evolution for a unique infinite-volume ground state in the domain of the infinite-volume Hamiltonian, 
and proved that there is no spectral gap above the ground state if a spontaneous symmetry breaking occurs. 

We write 
\begin{equation}
\label{AR0}
\mathcal{A}_R:=\frac{1}{|\Omega_R|}\sum_{x\in\Omega_R}(-1)^{x^{(1)}+\cdots+x^{(d)}}S_x^{(2)},
\end{equation}
where $\Omega_R$ is the hypercubic box with the side length $2R+1$ with positive integer $R$.   
\medskip

By improving Momoi's argument \cite{Momoi2}, we obtain:  

\begin{theorem}
\label{theorem:slowdecayZero}
Suppose that the spontaneous magnetization is non-vanishing 
for an infinite-volume ground state in dimensions $d\ge 2$. Then,  
the transverse correlation function $\omega_0( S_x^{(2)}S_y^{(2)})$ in the infinite-volume limit exhibits 
a Nambu-Goldstone-type slow decay. More precisely, we obtain 
\begin{equation}
\frac{|m_{\rm s}|^2}{R^{d-1}}\le {\rm Const.}\omega_0(\mathcal{A}_R^2).
\end{equation}
This bound rules out the possibility of the rapid decay $o(|x-y|^{-(d-1)})$ for the correlation, 
where $o(\varepsilon)$ denotes a quantity $q(\varepsilon)$ such that $q(\varepsilon)/\varepsilon$ is vanishing 
in the limit $\varepsilon\searrow 0$.
\end{theorem}

The proof is given in Sec.~\ref{sec:TransCorr}. 

%%%%%%%%%%%%%%%%%%%%%%%%%%%%%%%%%%%%%%%%%%
\subsection{Non-zero temperatures}

The thermal expectation value at the inverse temperature $\beta$ is given by 
\begin{equation}
\label{thermalstate}
\langle\cdots\rangle_{B,\beta}^{(\Lambda)}:=\frac{1}{Z_{B,\beta}^{(\Lambda)}}
{\rm Tr}\;(\cdots)e^{-\beta H_{\rm p}^{(\Lambda)}(B)},
\end{equation}
where $Z_{B,\beta}^{(\Lambda)}$ is the partition function. The infinite-volume thermal equilibrium state 
is given by 
\begin{equation}
\label{rhoB}
\rho_B(\cdots)={\rm weak}^\ast\mbox{-}\lim_{\Lambda\nearrow\ze^d}\langle\cdots\rangle_{B,\beta}^{(\Lambda)}.
\end{equation}
Here, if necessary, we take a suitable sequence of the finite lattices $\Lambda$ 
in the weak$^\ast$ limit so that the expectation value converges to a linear functional 
in the same way as in the case of the ground state (\ref{omegaPhi0B}). Similarly, we write 
\begin{equation}
\label{def:rho0}
\rho_0(\cdots):={\rm weak}^\ast\mbox{-}\lim_{B\searrow 0}\rho_B(\cdots).
\end{equation}
Then, the spontaneous magnetization $m_{{\rm s},\beta}$ is given by 
\begin{equation}
\label{msbeta}
m_{{\rm s},\beta}:=\lim_{\Lambda\nearrow\ze^d}\frac{1}{|\Lambda|}\rho_0(O^{(\Lambda)})
\end{equation}
with the hypercubic lattice $\Lambda$ of (\ref{Lambda}) with the even side length $2L$. 
Because of the translational invariance with period 2, the right-hand side dose not depend on 
the size of the lattice $\Lambda$. Clearly, this can be written 
$$
m_{{\rm s},\beta}=\lim_{B\searrow 0}
\lim_{\Lambda\nearrow\ze^d}\frac{1}{|\Lambda|}\langle O^{(\Lambda)}\rangle_{B,\beta}^{(\Lambda)}
$$
when the same sequences as in (\ref{rhoB}) and (\ref{def:rho0}) are taken in the double limit. 
Similarly to the case of the zero temperature, the value of $m_{{\rm s},\beta}$ may depend on 
the sequences in the double limit.
In general, let $\rho$ be a translationally invariant thermal equilibrium state (i.e., Gibbs state) 
which minimizes the free energy per volume \cite{BR}. Then, the spontaneous magnetization is given by 
$$
\frac{1}{|\Lambda|}\rho(O^{(\Lambda)})
$$
with the hypercubic lattice $\Lambda$ of (\ref{Lambda}) with the even side length $2L$.

\begin{theorem}
\label{thm:sponmagfiniteT}
For strictly positive temperatures $\beta^{-1}>0$, the value $m_{{\rm s},\beta}$
of (\ref{msbeta}) is independent of the choices of the sequences in the two weak$^\ast$-limits, 
$\Lambda\nearrow\ze^d$ and $B\searrow 0$, in (\ref{rhoB}) and (\ref{def:rho0}). Namely, 
the value $m_{{\rm s},\beta}$ of the spontaneous magnetization is uniquely determined. 
Further, the spontaneous magnetization $m_{{\rm s},\beta}$ is equal to 
the supremum of the spontaneous magnetizations 
over all of the translationally invariant thermal equilibrium states in the infinite volume.   
\end{theorem}

The proof is given in Sec.~\ref{Sec:SponMagFiniteT}.

\begin{theorem}
\label{theorem:slowdecayNonzero}
Suppose that the spontaneous magnetization is non-vanishing 
for an infinite-volume ground state for strictly positive temperatures $\beta^{-1}>0$ in dimensions $d\ge 3$.  
Then, the transverse correlation function $\rho_0( S_x^{(2)}S_y^{(2)})$ in the infinite-volume limit exhibits 
a Nambu-Goldstone-type slow decay. More precisely, we obtain 
\begin{equation}
\frac{|m_{{\rm s},\beta}|^2}{\beta R^{d-2}}\le{\rm Const.}\rho_0(\mathcal{A}_R^2), 
\end{equation}
where $\mathcal{A}_R$ is given by (\ref{AR0}). This bound rules out the possibility of 
the rapid decay $o(|x-y|^{-(d-2)})$ for the correlation. 
\end{theorem}

The proof is given in Sec.~\ref{sec:TransCorr}. 
\medskip

\noindent
Remark: The exponent $(d-2)$ is expected to be optimal because Kennedy and King \cite{KK} proved that 
the correlation corresponding to the Nambu-Goldstone mode exhibits exactly this power $(d-2)$ in the power-law decay 
in an Abelian Higgs model in Landau gauge in dimensions $d\ge 3$.  
By relying on Bogoliubov inequality, Martin \cite{Martin} discussed 
Nambu-Goldstone-type slow clustering of the transverse correlations 
when a continuous symmetry is broken in generic quantum or classical spin systems for strictly positive temperatures.  
See also a related approach \cite{LFPW} to proving Nambu-Goldstone-type slow clustering of the transverse correlations. 

%%%%%%%%%%%%%%%%%%%%%%%%%%%%%%%%%%%%%%%%%%%%%%%%%
\Section{Spontaneous Magnetization at Zero Temperature}
\label{sec:SMMaxZero}

In this section, we will prove the statement of Theorem~\ref{thm:sponmag} in the case of the ground states.  
For this purpose, we will use the theorem by Bratteli, Kishimoto and Robinson \cite{BKR} 
and the variational principle \cite{KHvdL} for the ground-state energy. 

Let $\Phi_0^{(\Lambda)}(B)$ be a ground-state vector of the finite-volume Hamiltonian $H_{\rm p}^{(\Lambda)}(B)$. 
Then, the expectation value of the ground-state energy satisfies   
$$
\langle \Phi_0^{(\Lambda)}(B),H_{\rm p}^{(\Lambda)}(B) \Phi_0^{(\Lambda)}(B)\rangle
\le \omega(H_{\rm p}^{(\Lambda)}(B)).
$$
for any state $\omega$. Substituting the right-hand side of (\ref{H(B)}) into this, one has  
$$
\langle \Phi_0^{(\Lambda)}(B),H_{0,{\rm p}}^{(\Lambda)}\Phi_0^{(\Lambda)}(B)\rangle 
-B\langle \Phi_0^{(\Lambda)}(B),O^{(\Lambda)}\Phi_0^{(\Lambda)}(B)\rangle\le 
\omega(H_{0,{\rm p}}^{(\Lambda)})-B\omega(O^{(\Lambda)}). 
$$
Further, this can be rewritten to the relation between two magnetizations as \cite{KHvdL} 
\begin{eqnarray}
\frac{1}{|\Lambda|}\langle \Phi_0^{(\Lambda)}(B),O^{(\Lambda)}\Phi_0^{(\Lambda)}(B)\rangle 
&\ge &\frac{1}{|\Lambda|}\omega(O^{(\Lambda)})+
\frac{1}{B|\Lambda|} [\langle \Phi_0^{(\Lambda)}(B),H_{0,{\rm p}}^{(\Lambda)}\Phi_0^{(\Lambda)}(B)\rangle 
-\omega(H_{0,{\rm p}}^{(\Lambda)})]\nonumber \\ 
&\ge&\frac{1}{|\Lambda|}\omega(O^{(\Lambda)})+
\frac{1}{B|\Lambda|} [E_0^{(\Lambda)}-\omega(H_{0,{\rm p}}^{(\Lambda)})],
\label{magneIneq}
\end{eqnarray}
where $E_0^{(\Lambda)}$ is the ground-state eigenenergy 
of the Hamiltonian, $H_{\rm p}^{(\Lambda)}(0)=H_{0,{\rm p}}^{(\Lambda)}$, with the external magnetic field $B=0$.

As the above state $\omega$, we choose an infinite-volume ground state which exhibits the maximum spontaneous 
magnetization for the Hamiltonian without the external magnetic field. 
Such a state is expected to be realized in an infinite-volume limit for a Hamiltonian 
with a boundary field which induces a spontaneous magnetization 
or a Hamiltonian with an infinitesimally weak symmetry breaking field which is switched off 
after taking the infinite-volume limit \cite{KomaTasaki}.

We define a set of states by 
$$
\mathcal{S}_\Lambda^\omega:=\{\omega'\bigr|\; \omega'|\mathfrak{A}_{\Lambda^c}=\omega|\mathfrak{A}_{\Lambda^c}\}
$$
for the state $\omega$ and a finite lattice $\Lambda$. Here, 
$\omega|\mathfrak{A}_{\Lambda^c}$ is the restriction of the state $\omega$ to the $C^\ast$-subalgebra $\mathfrak{A}_{\Lambda^c}$ 
on the complement $\Lambda^c$ of $\Lambda$. 
In order to estimate the energy difference in the right-hand side of (\ref{magneIneq}), we recall 
the theorem by Bratteli, Kishimoto and Robinson \cite{BKR}. 
We denote by $H_{0,{\rm f}}^{(\Lambda)}$ the Hamiltonian without the external magnetic field and 
with the free boundary condition on the finite lattice $\Lambda$. 
Then, the theorem states 

\begin{theorem}{\bf (Bratteli, Kishimoto and Robinson \cite{BKR})} 
Let $\omega$ be an infinite-volume ground state $\omega$, i.e., $\omega$ satisfies the ground-state 
condition (\ref{defGS}). 
Then, the state $\omega$ satisfies 
\begin{equation}
\label{BKRcharaGS}
\omega(\tilde{H}_0^{(\Lambda)})=\inf_{\omega'\in \mathcal{S}_\Lambda^\omega}\omega'(\tilde{H}_0^{(\Lambda)}),
\end{equation}
where 
$$
\tilde{H}_0^{(\Lambda)}=H_{0,{\rm f}}^{(\Lambda)}+W^{(\Lambda)} 
$$
with the boundary Hamiltonian,  
$$
W^{(\Lambda)}:=\sum_{X:X\cap\Lambda\neq \emptyset,X\cap\Lambda^c\not=\emptyset}h_X.
$$
Here, $h_X$ is the local Hamiltonian on $X\subset\ze^d$.
\end{theorem}

As a trial state, we take 
\begin{equation}
\label{tildeomega}
\tilde{\omega}(\cdots):=\omega\bigr|\mathfrak{A}_{\Lambda^c}(\cdots)
\otimes \langle \Phi_0^{(\Lambda)}(0),(\cdots)\Phi_0^{(\Lambda)}(0)\rangle
\end{equation}
by using the finite-volume ground state $\Phi_0^{(\Lambda)}(0)$ of the Hamiltonian 
$H_{\rm p}^{(\Lambda)}(0)=H_{0,{\rm p}}^{(\Lambda)}$ with zero external magnetic field $B=0$. 

Then, from (\ref{BKRcharaGS}), the following bound holds:  
\begin{equation}
\label{BKRbound}
\tilde{\omega}(\tilde{H}_0^{(\Lambda)})\ge \omega(\tilde{H}_0^{(\Lambda)})
\end{equation}
because the state $\omega$ is an infinite-volume ground state. 
We write 
$$
\delta H_0^{(\Lambda)}:=\tilde{H}_0^{(\Lambda)}-H_{0,{\rm p}}^{(\Lambda)}. 
$$
Substituting this and the trial state (\ref{tildeomega}) into the bound (\ref{BKRbound}), one has 
$$
E_0^{(\Lambda)}+\tilde{\omega}(\delta{H}_0^{(\Lambda)})\ge \omega({H}_{0,{\rm p}}^{(\Lambda)})
+\omega(\delta{H}_0^{(\Lambda)}).
$$
Since $\Vert \delta{H}_0^{(\Lambda)}\Vert\le {\rm Const.}L^{d-1}$, this yields 
$$
\omega({H}_{0,{\rm p}}^{(\Lambda)})-E_0^{(\Lambda)}\le {\rm Const.}L^{d-1}.
$$
Substituting this into the right-hand side of (\ref{magneIneq}) 
and using the expression (\ref{ms}) of $m_{\rm s}$, we obtain 
\begin{equation}
\label{m}
m_{\rm s}=\lim_{B\searrow 0}\lim_{\Lambda\nearrow\ze^d}
\frac{1}{|\Lambda|}\langle \Phi_0^{(\Lambda)}(B),O^{(\Lambda)}\Phi_0^{(\Lambda)}(B)\rangle 
\ge \lim_{\Lambda\nearrow\ze^d}\frac{1}{|\Lambda|}\omega(O^{(\Lambda)}).
\end{equation}
Here, if necessary, we take a suitable sequence of finite lattices $\Lambda$ going to $\ze^d$ 
so that the right-hand side yields the maximum magnetization for the infinite-volume ground state $\omega$, 
and we choose the sequence of the limit $B\searrow 0$ so that the limit converges to some value $m_{\rm s}$ 
in the set of all of the accumulation points.  
Further, consider a state, 
$$
\omega_{\Phi_0(+0)}:={\rm weak}^\ast\mbox{-}\lim_{B\searrow 0}\omega_{\Phi_0(B)}, 
$$ 
where we take the sequence of the weak$^\ast$-limit $B\searrow 0$ so that 
$$
m_{\rm s}=\lim_{B\searrow 0}\frac{1}{|\Lambda|}\omega_{\Phi_0(B)}(O^{(\Lambda)})
=\frac{1}{|\Lambda|}\omega_{\Phi_0(+0)}(O^{(\Lambda)})
$$
with the same value $m_{\rm s}$ as in above {from} (\ref{mPhi0B}), (\ref{ms0}), (\ref{omegaPhi00}) and (\ref{ms}). 
As we explained in Sec.~\ref{ZeroTempThms}, the state $\omega_{\Phi_0(+0)}$ is an infinite-volume ground state 
in the case of the external magnetic field $B=0$. 
Since the right-hand side of (\ref{m}) gives the maximum spontaneous magnetization 
over all of the infinite-volume ground states from the assumption, we have 
$$
m_{\rm s}\ge \lim_{\Lambda\nearrow\ze^d}\frac{1}{|\Lambda|}\omega(O^{(\Lambda)})
\ge \lim_{\Lambda\nearrow\ze^d}\frac{1}{|\Lambda|}\omega_{\Phi_0(+0)}(O^{(\Lambda)})=m_{\rm s}.
$$
This implies that the above value $m_{\rm s}$ is always equal to 
the maximum magnetization for the state $\omega$, irrespective of the sequence of the limit $B\searrow 0$. 
In particular, when the state $\omega$ is translationally invariant with period 2, 
the quantity $\omega(O^{(\Lambda)})/|\Lambda|$ does not depend on the size of the lattice $\Lambda$. 
Therefore, we do not need to choose the above suitable sequence of finite lattices $\Lambda$ 
for the inequality (\ref{m}). 
This implies that $m_{\rm s}$ is equal to the maximum spontaneous magnetization over all of 
the translationally invariant infinite-volume ground states, irrespective of the sequences of 
finite lattices $\Lambda$ going to $\ze^d$ in the weak${^\ast}$-limit 
for the state $\omega_{\Phi_0(B)}$ of (\ref{omegaPhi0B}).  
Moreover, if we choose the appropriate sequence of finite lattices $\Lambda$ going to $\ze^d$ 
so that the right-hand side of the above inequality (\ref{m}) yields the maximum magnetization 
for the infinite-volume ground state $\omega$, similarly we have that 
$m_{\rm s}$ is equal to the maximum spontaneous magnetization over all of 
the infinite-volume ground states which are not necessarily translationally invariant. 
This completes the proof of Theorem~\ref{thm:sponmag}. 

%%%%%%%%%%%%%%%%%%%%%%%%%%%%%%%%%%%%%%%%%
\Section{Spontaneous Magnetization at Non-zero Temperatures}
\label{Sec:SponMagFiniteT}

In this section, we will prove the statement of Theorem~\ref{thm:sponmagfiniteT}. 
Namely, we prove that the spontaneous magnetization $m_{{\rm s},\beta}$ of (\ref{msbeta}) is equal to 
the maximum spontaneous magnetization over all of the translationally invariant thermal equilibrium states 
for strictly positive temperatures $\beta^{-1}>0$. 
For this purpose, we will use the fact that 
thermal equilibrium states (i.e., Gibbs states) minimizes the free energy per volume \cite{BR} 
for translationally invariant systems.   

Consider a Hamiltonian $H(B)$ which is written 
as a sum of the translates of the local Hamiltonian $h_x(B):=h_x-Bo_x$, where $B\ge 0$ is an external magnetic field 
and $o_x$ is the local order parameter. Namely, the Hamiltonian $H(B)$ is translationally invariant and formally written as 
$$
H(B)=\sum_x (h_x-Bo_x).
$$

Let $\rho$ be a translationally invariant infinite-volume state. We write 
$$
e_B(\rho)=\rho((h_x-Bo_x))\quad \mbox{and} \quad e_0(\rho)=\rho(h_x)
$$
for the expectation value of the energy per volume with $B>0$ and $B=0$, respectively. 
For the state $\rho$ and the subalgebra $\mathfrak{A}_\Lambda$ on the finite lattice $\Lambda$, 
the density matrix $\sigma_\Lambda$ is uniquely determined by \cite{BR}
$$
\rho(a)={\rm Tr}_{\mathfrak{H}_\Lambda}(\sigma_\Lambda a)
$$
for all $a\in\mathfrak{A}_\Lambda$. Then, the entropy per volume in the infinite volume is given by \cite{BR} 
$$
s(\rho):=- \lim_{\Lambda\nearrow\ze^d}\frac{1}{|\Lambda|}{\rm Tr}_{\mathfrak{H}_\Lambda}(\sigma_\Lambda
\log \sigma_\Lambda).
$$  
This limit is known to exist \cite{BR}. 

Let $\rho_B$ be an infinite-volume Gibbs state \cite{BR} at the inverse temperature $\beta$ and 
for the external magnetic field $B$, i.e., 
the state $\rho_B$ minimizes the free energy or equivalently maximizes $s(\rho_B)-\beta e_B(\rho_B)$. 
Therefore, the following inequality holds: 
$$
s(\rho_B)-\beta e_B(\rho_B)\ge s(\rho)-\beta e_B(\rho)=s(\rho)-\beta e_0(\rho)+\beta B\rho(o_x)
$$
for any translationally invariant infinite-volume state $\rho$. 
We choose $\rho$ to be a Gibbs state for the Hamiltonian $H(0)$ with the zero external field $B=0$. 
Therefore, one has 
$$
s(\rho)-\beta e_0(\rho)\ge s(\rho_B)-\beta e_0(\rho_B).
$$
Combining these two inequalities, one obtains 
$$
\rho_B(o_x)\ge \rho(o_x).
$$
In particular, when the state $\rho$ yields the maximum magnetization $m_{{\rm max},\beta}$, we have 
$$
m_{{\rm s},\beta}= \lim_{B\searrow 0}\rho_B(o_x)\ge m_{{\rm max},\beta},
$$
where we take the sequence of the limit $B\searrow 0$ so that the limit converges to some value $m_{{\rm s},\beta}$ 
in the set of all of the accumulation points. 
Consider 
$$
\rho_0(\cdots)={\rm weak}^\ast\mbox{-}\lim_{B\searrow 0}\rho_B(\cdots),
$$
where we choose the sequence of the limit $B\searrow 0$ so that $\rho_0(o_x)=m_{{\rm s},\beta}$, i.e., 
the state shows the same spontaneous magnetization as in above. Then, one can easily check that 
the state $\rho_0$ is a translationally invariant thermal equilibrium states.  
Therefore, we have 
$$
m_{{\rm s},\beta}\ge m_{{\rm max},\beta}\ge m_{{\rm s},\beta}. 
$$
This implies that the spontaneous magnetization $m_{{\rm s},\beta}$ is equal to the maximum magnetization 
for strictly positive temperatures $\beta^{-1}>0$, too. 
Thus, the statement of Theorem~\ref{thm:sponmagfiniteT} has been proved.

%%%%%%%%%%%%%%%%%%%%%%%%%%%%%%%%%%%%%%%%%%%
\Section{A Trial State for Low-Energy Excitations}
\label{sec:NambuGoldstone}

In this section, we give a proof of Theorem~\ref{thm:GSgapless}. 
In order to prove the existence of a gapless excitation above the sector of the ground state, 
we will basically use the variational principle with respect to energy. Our key idea is to choose 
the trial state to be a special form (\ref{trialvarphi}) below that eliminates the contributions of 
the undesired low-lying eigenstates, which yield a set of symmetry-breaking ground states in the infinite-volume limit 
by forming linear combinations of the low-lying eigenstates and the symmetric ground state, as mentioned in 
Introduction.

We denote by $E_0^{(\Lambda)}(B)$ the eigenenergy of the ground state of the Hamiltonian $H_{\rm p}^{(\Lambda)}(B)$ 
in the external magnetic field with the strength $B$. 
We write 
$$
\mathcal{H}_{\rm p}^{(\Lambda)}(B):=H_{\rm p}^{(\Lambda)}(B)-E_0^{(\Lambda)}(B).
$$
For an operator $\mathcal{A}\in\mathfrak{A}_\Lambda$, we introduce a trial state as 
\begin{equation}
\label{trialvarphi}
\varphi_{B,\epsilon,\mathcal{A}}^{(\Lambda)}(\cdots):=
\frac{\omega_B^{(\Lambda)}(\mathcal{A}^\ast [\mathcal{H}_{\rm p}^{(\Lambda)}(B)]^{\epsilon/2} 
(\cdots)[\mathcal{H}_{\rm p}^{(\Lambda)}(B)]^{\epsilon/2}\mathcal{A})}
{\omega_B^{(\Lambda)}(\mathcal{A}^\ast [\mathcal{H}_{\rm p}^{(\Lambda)}(B)]^\epsilon \mathcal{A})},
\end{equation}
where the ground state $\omega_B^{(\Lambda)}$ is given by (\ref{Grstate}), and 
$\epsilon$ is a positive small parameter. Here, we stress the following again: 
If an excitation energy appeared in (\ref{trialvarphi}) 
is very close to the energy $E_0^{(\Lambda)}(B)$ of the ground state, then the contribution becomes very small 
due to the factor $[\mathcal{H}_{\rm p}^{(\Lambda)}(B)]^{\epsilon/2}$. Thus, we can eliminate the contributions of 
the undesired low-lying eigenstates. But, for the purpose of giving the proof of Theorem~\ref{thm:GSgapless} 
in a mathematically rigorous way, 
we slightly modify the cutoff $[\mathcal{H}_{\rm p}^{(\Lambda)}(B)]^{\epsilon/2}$ 
in Sec.~\ref{ProofTheoremGSgapless} below. 

The energy expectation with respect to the state (\ref{trialvarphi}) is given by 
\begin{equation}
\label{lowenergy}
\varphi_{B,\epsilon,\mathcal{A}}^{(\Lambda)}(\mathcal{H}_{\rm p}^{(\Lambda)}(B))
=\frac{\omega_B^{(\Lambda)}(\mathcal{A}^\ast [\mathcal{H}_{\rm p}^{(\Lambda)}(B)]^{1+\epsilon}\mathcal{A})}
{\omega_B^{(\Lambda)}(\mathcal{A}^\ast [\mathcal{H}_{\rm p}^{(\Lambda)}(B)]^\epsilon \mathcal{A})}. 
\end{equation}
Let $R$ be a large positive integer, and define  
\begin{equation}
\label{OmegaR}
\Omega_R:=\{x\in \ze^d\; |\; |x|_\infty\le R\}\subset\ze^d, 
\end{equation}
where 
$$
|x|_\infty:=\max_{1\le i\le d}\{|x^{(i)}|\}\quad \mbox{for \ } x=(x^{(1)},x^{(2)},\ldots,x^{(d)})\in \ze^d. 
$$
We choose the local operator $\mathcal{A}\in\mathfrak{A}_{\rm loc}$ as 
\begin{equation}
\label{AR}
\mathcal{A}=\mathcal{A}_R:=\frac{1}{|\Omega_R|}\sum_{x\in\Omega_R}(-1)^{x^{(1)}+\cdots+x^{(d)}}S_x^{(2)}. 
\end{equation}
Clearly, ${\cal A}_R\in\mathfrak{A}_{\Omega_R}$.

%%%%%%%%%%%%%%%%%%%%%%%%%%%%%%%%%%%%%%%%%%%%%%%%%%%%%%%%%%%%%%%%%%%%%%%%%%%%%%%%%%%%%%%%
\subsection{Estimate of the denominator of the right-hand side in (\ref{lowenergy})}

In order to estimate the denominator of the right-hand side in (\ref{lowenergy}), 
we use the following Kennedy-Lieb-Shastry type inequality \cite{KLS}:

\begin{lemma}
\label{lem:BogolyIneq}
Let $A, C$ be operators on $\Lambda$, and let $\epsilon$ be a positive small parameter. Then, the following bound is valid: 
\begin{eqnarray}
\label{BogolyIneq}
|\omega_B^{(\Lambda)}([C,A])|^2
&\le& \sqrt{\tilde{D}_B^{(\Lambda)}(C)}
\sqrt{\kappa(\epsilon)\;\omega_B^{(\Lambda)}(\{C,C^\ast\})+\omega_B^{(\Lambda)}([[C^\ast,H_{\rm p}^{(\Lambda)}(B)],C])}\nonumber \\
&\times& \Bigl\{\omega_B^{(\Lambda)}(A[\mathcal{H}_{\rm p}^{(\Lambda)}(B)]^\epsilon A^\ast)
+\omega_B^{(\Lambda)}(A^\ast [\mathcal{H}_{\rm p}^{(\Lambda)}(B)]^\epsilon A)\Bigr\},
\end{eqnarray}
where  
\begin{equation}
\tilde{D}_B^{(\Lambda)}(C):=
\omega_B^{(\Lambda)}(CP_{\rm ex}^{(\Lambda)}(B)[\mathcal{H}_{\rm p}^{(\Lambda)}(B)]^{-1}C^\ast) 
+\omega_B^{(\Lambda)}(C^\ast P_{\rm ex}^{(\Lambda)}(B)[\mathcal{H}_{\rm p}^{(\Lambda)}(B)]^{-1}C)
\label{tildeD}
\end{equation}
with
\begin{equation} 
P_{\rm ex}^{(\Lambda)}(B):=1-P_0^{(\Lambda)}(B),
\end{equation}
and $\kappa(\epsilon)$ is a positive function of the parameter $\epsilon$ such that $\kappa(\epsilon)\rightarrow 0$ as 
$\epsilon\rightarrow 0$. Here, $P_0^{(\Lambda)}(B)$ is the projection onto the sector of the ground state of 
the Hamiltonian $H_{\rm p}^{(\Lambda)}(B)$. 
\end{lemma}

The proof is given in Appendix~\ref{App::BogolyIneq}. The inequality (\ref{BogolyIneq}) is a slight extension of the 
Kennedy-Lieb-Shastry inequality \cite{KLS}. In fact, when the parameter $\epsilon$ is zero, it is nothing but 
the Kennedy-Lieb-Shastry inequality \cite{Momoi2}. 
\bigskip

First, in order to estimate $\tilde{D}_B^{(\Lambda)}(C)$ in the right-hand side of the inequality (\ref{BogolyIneq}), 
we use the inequality (\ref{chibound}) below which is obtained 
by relying on the reflection positivity of the model \cite{KLS}.

\begin{lemma}
Let $f=f_x$ for $x\in\Lambda$ be a real-valued function on the lattice $\Lambda$. Then, the following bound is valid: 
\begin{equation}
\label{chibound}
\omega_B^{(\Lambda)}(H_1'P_{\rm ex}^{(\Lambda)}(B)[\mathcal{H}_{\rm p}^{(\Lambda)}(B)]^{-1}H_1')
\le \sum_{\{x,y\}\subset\Lambda:|x-y|=1}\frac{1}{2}(f_x+f_y)^2,
\end{equation}
where 
\begin{equation}
\label{H1'}
H_1'=\sum_{\{x,y\}\subset\Lambda:|x-y|=1}\left[S_x^{(1)}+S_y^{(1)}\right](f_x+f_y).
\end{equation} 
\end{lemma}

\begin{proof}{Proof}
Consider a Hamiltonian, 
\begin{eqnarray}
\label{HamBf}
H_{\rm p}^{(\Lambda)}(B,f)&:=&\sum_{\{x,y\}\subset\Lambda:|x-y|=1}\left[S_x^{(2)}S_y^{(2)}+S_x^{(3)}S_y^{(3)}\right]
-BO^{(\Lambda)}\nonumber \\
&+&\frac{1}{2}\sum_{\{x,y\}\subset\Lambda:|x-y|=1}\left[(S_x^{(1)}+S_y^{(1)}+f_x+f_y)^2-(S_x^{(1)})^2-(S_y^{(1)})^2\right].
\end{eqnarray}
Clearly, this is written 
\begin{equation}
H_{\rm p}^{(\Lambda)}(B,f)=H_{\rm p}^{(\Lambda)}(B)
+\sum_{\{x,y\}\subset\Lambda:|x-y|=1}\left[S_x^{(1)}+S_y^{(1)}\right](f_x+f_y)
+\sum_{\{x,y\}\subset\Lambda:|x-y|=1}\frac{1}{2}(f_x+f_y)^2,
\end{equation}
where the first term in the right-hand side is given by (\ref{H(B)}). 
Applying the method of the reflection positivity yields that the ground-state energy $E_0^{(\Lambda)}(B,f)$ of 
the Hamiltonian $H_{\rm p}^{(\Lambda)}(B,f)$ satisfies \cite{KLS} 
\begin{equation}
\label{E0bound}
E_0^{(\Lambda)}(B,f)\ge E_0^{(\Lambda)}(B,0)=E_0^{(\Lambda)}(B)\ \ \mbox{for any \ } f.
\end{equation}
Namely, the minimum value of the ground-state energy is given by $f=0$. 
The proof is given in Appendix~\ref{RP}. From this inequality, one has  
the bound (\ref{chibound}). The proof is given in Appendix~\ref{appen:chibound}.
\end{proof} 

For the inequality (\ref{BogolyIneq}), we choose $A=\mathcal{A}_R$ of (\ref{AR}) and $C=H_1'$ of (\ref{H1'}) 
with the function $f_x$ which is given by 
\begin{equation}
\label{fx}
f_x=\cases{1, & $x\in\Omega_{R+1}$;\cr
1-[|x|_\infty-(R+1)]/R, & $x\in\Omega_{2R}\backslash\Omega_{R+1}$;\cr
0, & otherwise,}
\end{equation}
where $\Omega_R$ is given by (\ref{OmegaR}) with the positive integer $R$. Then, one has 
\begin{equation}
\label{estD}
\tilde{D}_B^{(\Lambda)}(C)=\tilde{D}_B^{(\Lambda)}(H_1')\le {\rm Const.}R^d
\end{equation}
{from} the definition (\ref{tildeD}) of $\tilde{D}_B^{(\Lambda)}$ and the inequality (\ref{chibound}).

The commutator $[C,A]=[H_1',\mathcal{A}_R]$ in the left-hand side of (\ref{BogolyIneq}) is calculated as 
\begin{eqnarray}
\label{CommuCA}
[C,A]=[H_1',\mathcal{A}_R]&=&\frac{1}{|\Omega_R|}\sum_{x\in\Omega_R}(-1)^{x^{(1)}+\cdots+x^{(d)}}[H_1',S_x^{(2)}]\nonumber\\
&=&\frac{2d}{|\Omega_R|}\sum_{x\in\Omega_R}(-1)^{x^{(1)}+\cdots+x^{(d)}}iS_x^{(3)}
=\frac{2di}{|\Omega_R|}O^{(\Omega_R)}.
\end{eqnarray}
Therefore, one has 
\begin{equation}
\label{omegaBLambda[C,A]}
\omega_B^{(\Lambda)}([C,A])=\frac{2di}{|\Omega_R|}\omega_B^{(\Lambda)}(O^{(\Omega_R)}).
\end{equation}
This expectation value is nothing but the staggered magnetization which is nonvanishing \cite{KHvdL,KomaTasaki} for 
taking the infinite-volume limit $\Lambda\nearrow\ze^d$ first and then the zero field limit $B\searrow 0$. 
We write 
\begin{equation}
\label{msLambdaB}
m_{\rm s}^{(\Lambda)}(B):=\frac{1}{|\Omega_R|}\omega_B^{(\Lambda)}(O^{(\Omega_R)}).
\end{equation}
By using the lattice shift and the rotation by $\pi$ about the 2 axis in the spin space, one can easily 
show 
\begin{equation}
(-1)^{x^{(1)}+\cdots+x^{(d)}}\omega_B^{(\Lambda)}(S_x^{(3)})
=(-1)^{y^{(1)}+\cdots+y^{(d)}}\omega_B^{(\Lambda)}(S_y^{(3)})
\end{equation}
for even $(x^{(1)}+\cdots+x^{(d)})$ and odd $(y^{(1)}+\cdots+y^{(d)})$. 
We recall that the infinite-volume ground state $\omega_0$ of (\ref{omega0}) exhibits 
the maximum spontaneous magnetization $m_{\rm s}$. From these observations, we obtain 
\begin{equation}
m_{\rm s}=\frac{1}{|\Omega_R|}\omega_0(O^{(\Omega_R)}).
\end{equation}
Namely, 
\begin{equation}
\label{mslimitzero}
m_{\rm s}=\lim_{B\searrow 0}\lim_{\Lambda\nearrow \ze^d}m_{\rm s}^{(\Lambda)}(B)
\end{equation}
if the same sequences as in (\ref{omega0}) are taken in the double limit. 

Further, one obtains 
\begin{eqnarray}
\label{estH1SxSyH1}
[[C^\ast,H_{\rm p}^{(\Lambda)}(B)],C]&=&[[C^\ast,H_{0,{\rm p}}^{(\Lambda)}],C]
-B[[C^\ast,O^{(\Lambda)}],C] \nonumber \\
&=&\sum_{\{x,y\}\subset\Omega_{2R+1}:|x-y|=1}[[H_1',{\bf S}_x\cdot{\bf S}_y],H_1']-B[[H_1',O^{(\Omega_{2R+1})}],H_1']
\nonumber \\ 
\end{eqnarray}
for the double commutator in the right-hand side of (\ref{BogolyIneq}). 
The sum in the right-hand side can be estimated by using the U(1) symmetry of the Hamiltonian 
(rotation around the $1$-axis in the spin space). This symmetry reduces the order of the quantity 
{from} ${\cal O}(R^d)$ to ${\cal O}(R^{d-2})$ as follows: 

\begin{lemma}
\begin{equation}
\label{estsumH1SxSyH1}
\sum_{\{x,y\}\subset\Omega_{2R+1}:|x-y|=1}\Vert[[H_1',{\bf S}_x\cdot{\bf S}_y],H_1']\Vert
\le {\rm Const.} R^{d-2}. 
\end{equation}
\end{lemma}

\begin{proof}{Proof} {From} the definition (\ref{H1'}) of $H_1'$, one has 
\begin{equation}
\label{doublecommuH1SS}
[[H_1',{\bf S}_x\cdot{\bf S}_y],H_1']
=4[[\tilde{S}_x^{(1)}+\tilde{S}_y^{(1)},{\bf S}_x\cdot{\bf S}_y],\tilde{S}_x^{(1)}+\tilde{S}_y^{(1)}],
\end{equation}
where 
$$
\tilde{S}_x^{(1)}:=\sum_{y':|x-y'|=1}S_x^{(1)}(f_x+f_{y'}).
$$
Note that 
$$
\tilde{S}_x^{(1)}-4df_xS_x^{(1)}=\sum_{y':|x-y'|=1}S_x^{(1)}(f_{y'}-f_x)=:\Delta S_x^{(1)}
$$
and 
$$
\tilde{S}_y^{(1)}-4df_xS_y^{(1)}=\sum_{x':|x'-y|=1}S_y^{(1)}(f_y-f_x+f_{x'}-f_x)=:\Delta S_y^{(1)}
$$
for $y$ satisfying $|x-y|=1$. Substituting these into the right-hand side of the above double 
commutator (\ref{doublecommuH1SS}), one has 
\begin{eqnarray*}
[[H_1',{\bf S}_x\cdot{\bf S}_y],H_1']
&=&16d f_x[[\Delta S_x^{(1)}+\Delta S_y^{(1)},{\bf S}_x\cdot{\bf S}_y],S_x^{(1)}+S_y^{(1)}]\nonumber\\
&+&4[[\Delta S_x^{(1)}+\Delta S_y^{(1)},{\bf S}_x\cdot{\bf S}_y],\Delta S_x^{(1)}+\Delta S_y^{(1)}].
\end{eqnarray*} 
Here, one can easily show that the first term in the right-hand side is vanishing. 
{From} the definitions of the function $f_x$, $\Delta S_x^{(1)}$ and $\Delta S_y^{(1)}$,  
the norms of these operators can be estimated as 
$$
\Vert \Delta S_x^{(1)}\Vert \le {\rm Const.}\frac{1}{R} \quad \mbox{and} \quad 
\Vert \Delta S_y^{(1)}\Vert \le {\rm Const.}\frac{1}{R}.
$$
Therefore, 
$$
\Vert[[H_1',{\bf S}_x\cdot{\bf S}_y],H_1']\Vert\le{\rm Const.}\frac{1}{R^2}.
$$
Consequently, one has the bound (\ref{estsumH1SxSyH1}). 
\end{proof}

{From} (\ref{estsumH1SxSyH1}) and (\ref{estH1SxSyH1}), one has 
\begin{equation}
\label{estsumH1SxSyH1p}
|\omega_B^{(\Lambda)}([[C^\ast,H_{\rm p}^{(\Lambda)}(B)],C])|
\le {\rm Const.}R^{d-2}+{\rm Const.}|B|R^d
\end{equation}
with $C=H_1'$. 

Now let us estimate the denominator of the expectation value (\ref{lowenergy}) of the excitation energy. 
Combining (\ref{BogolyIneq}), (\ref{estD}), (\ref{omegaBLambda[C,A]}), (\ref{msLambdaB}) 
and (\ref{estsumH1SxSyH1p}), we obtain 
\begin{equation}
\label{denomibound}
|m_{\rm s}^{(\Lambda)}(B)|^2\le\mathcal{K}_0R^{d-1}\Bigl[1+\mathcal{K}_1\kappa(\epsilon)R^{d+2}
+\mathcal{K}_2|B|R^2\Bigr]^{1/2}\times\omega_B^{(\Lambda)}(\mathcal{A}_R
[\mathcal{H}_{\rm p}^{(\Lambda)}(B)]^\epsilon\mathcal{A}_R),
\end{equation}
where $\mathcal{K}_0$, $\mathcal{K}_1$ and $\mathcal{K}_2$ are a positive constant. 

%%%%%%%%%%%%%%%%%%%%%%%%%%%%%%%%%%%%%%%%%%%%%%%%%%%%%%%%
\subsection{Estimate of the numerator in the right-hand side of (\ref{lowenergy})}

Next let us estimate the numerator of (\ref{lowenergy}).
For the Hamiltonian $\mathcal{H}_{\rm p}^{(\Lambda)}(B)$, we denote by $P(E',+\infty)$ the spectral projection 
onto the energies which are larger than $E'>0$. We also write $P[0,E'):=1-P(E',+\infty)$. 
Note that
\begin{eqnarray*}
& &\omega_B^{(\Lambda)}(\mathcal{A}_R[\mathcal{H}_{\rm p}^{(\Lambda)}(B)]^{1+\epsilon}\mathcal{A}_R)\\
&=&\omega_B^{(\Lambda)}(\mathcal{A}_RP[0,E')[\mathcal{H}_{\rm p}^{(\Lambda)}(B)]^{1+\epsilon}\mathcal{A}_R)
+\omega_B^{(\Lambda)}(\mathcal{A}_RP(E',+\infty)[\mathcal{H}_{\rm p}^{(\Lambda)}(B)]^{1+\epsilon}\mathcal{A}_R)\\
&\le&\omega_B^{(\Lambda)}(\mathcal{A}_RP[0,E')\mathcal{H}_{\rm p}^{(\Lambda)}(B)\mathcal{A}_R)
\times(E')^{\epsilon}+\omega_B^{(\Lambda)}(\mathcal{A}_RP(E',+\infty)[\mathcal{H}_{\rm p}^{(\Lambda)}(B)]^3\mathcal{A}_R)
\times (E')^{\epsilon-2}\\
&\le&\omega_B^{(\Lambda)}(\mathcal{A}_R\mathcal{H}_{\rm p}^{(\Lambda)}(B)\mathcal{A}_R)
\times(E')^{\epsilon}
+\omega_B^{(\Lambda)}(\mathcal{A}_R[\mathcal{H}_{\rm p}^{(\Lambda)}(B)]^3\mathcal{A}_R)\times (E')^{\epsilon-2}.
\end{eqnarray*}
The first term in the right-hand side in the last line can be estimated as 
\begin{equation}
\label{omegaBAHAexpvalue}
\omega_B^{(\Lambda)}(\mathcal{A}_R\mathcal{H}_{\rm p}^{(\Lambda)}(B)\mathcal{A}_R)
=\frac{1}{2}\omega_B^{(\Lambda)}([\mathcal{A}_R,[H_{\rm p}^{(\Lambda)}(B),\mathcal{A}_R]])
\le\frac{{\cal K}_3}{R^d}. 
\end{equation}
with the positive constant ${\cal K}_3$. Similarly, the second term is evaluated as 
\begin{eqnarray}
\label{omegaARPEPH3AR}
\omega_B^{(\Lambda)}(\mathcal{A}_R[\mathcal{H}_{\rm p}^{(\Lambda)}(B)]^3\mathcal{A}_R)
&\le&\omega_B^{(\Lambda)}(\mathcal{A}_R[\mathcal{H}_{\rm p}^{(\Lambda)}(B)]^3\mathcal{A}_R)\nonumber\\
&=&\omega_B^{(\Lambda)}([\mathcal{A}_R,H_{\rm p}^{(\Lambda)}(B)]\mathcal{H}_{\rm p}^{(\Lambda)}(B)
[H_{\rm p}^{(\Lambda)}(B),\mathcal{A}_R])\nonumber\\
&=&\omega_B^{(\Lambda)}(\mathcal{B}_R\mathcal{H}_{\rm p}^{(\Lambda)}(B)\mathcal{B}_R)\nonumber\\
&=&\frac{1}{2}\omega_B^{(\Lambda)}([\mathcal{B}_R,[H_{\rm p}^{(\Lambda)}(B),\mathcal{B}_R]])
\le\frac{{\cal K}_4}{R^d}
\end{eqnarray}
with the positive constant ${\cal K}_4$, 
where we have written $\mathcal{B}_R:=i[H_{\rm p}^{(\Lambda)}(B),\mathcal{A}_R]$, and used the assumption that 
the interactions of the Hamiltonian $H_{\rm p}^{(\Lambda)}(B)$ are of finite range.  
{From} these observations, we obtain 
\begin{equation}
\label{finalnumebound}
\omega_B^{(\Lambda)}(\mathcal{A}_R[\mathcal{H}_{\rm p}^{(\Lambda)}(B)]^{1+\epsilon}\mathcal{A}_R)
\le \frac{1}{R^d}\left[\mathcal{K}_3(E')^\epsilon+\mathcal{K}_4(E')^{\epsilon-2}\right]
\le \frac{{\cal K}_3+{\cal K}_4}{R^d},
\end{equation}
where we have chosen $E'=1$. 

%%%%%%%%%%%%%%%%%%%%%%%%%%%%%%%%%%%%%%%%%%%%%%%%%%%%%%%%%%%%%%%%%%%%
\subsection{Proof of Theorem~\ref{thm:GSgapless}}
\label{ProofTheoremGSgapless}

For the parameter, $\epsilon$, we choose $\epsilon$ which satisfies 
$$
\kappa(\epsilon)R^{d+2}\le 1 
$$
for a fixed $R$. Then, the inequality (\ref{denomibound}) is written 
\begin{equation}
\label{finaldenomibound}
|m_{\rm s}^{(\Lambda)}(B)|^2\le\mathcal{K}_0R^{d-1}\Bigl[1+\mathcal{K}_1
+\mathcal{K}_2|B|R^2\Bigr]^{1/2}\times\omega_B^{(\Lambda)}(\mathcal{A}_R
[\mathcal{H}_{\rm p}^{(\Lambda)}(B)]^\epsilon\mathcal{A}_R).
\end{equation}
{From} this and (\ref{finalnumebound}), we can obtain the desired bound 
$$
\varphi_{B,\epsilon,\mathcal{A}}^{(\Lambda)}(\mathcal{H}_{\rm p}^{(\Lambda)}(B))
\le \frac{\mathcal{K}_0({\cal K}_3+{\cal K}_4)}{R|m_{\rm s}^{(\Lambda)}(B)|^2}
\Bigl[1+\mathcal{K}_1+\mathcal{K}_2|B|R^2\Bigr]^{1/2} 
$$
for the expectation value of (\ref{lowenergy}). In the double limit $B\searrow 0$ and 
$\Lambda\nearrow\ze^d$, the excitation energy is bounded by ${\rm Const.}/R$ because 
the spontaneous magnetization is non-vanishing. 
Although this is the desired result, the cutoff 
$[\mathcal{H}_{\rm p}^{(\Lambda)}(B)]^{\epsilon/2}$ in the expression (\ref{lowenergy}) 
is slightly singular at the zero energy. 
Therefore, we approximate the function $(\cdots)^{\epsilon/2}$ with an infinitely differentiable function on $\re$ 
with compact support, i.e., a function in $C_0^\infty(\re)$ . 

To begin with, we extend the function $s^{\epsilon/2}$ for $s\ge 0$ to that for $s<0$ as follows: 
$$
\eta(s):=\cases{s^{\epsilon/2} & for $s\ge 0$; \cr 
                0 & for $s<0$. \cr}
$$
Next, we introduce $\hat{g}_1\in C^\infty(\re)$ which satisfies the conditions, 
$$
\hat{g}_1(s)=\cases{1 & for $s\le \gamma_1$; \cr 
                0 & for $s\ge \gamma_2$, \cr}
$$
and $0\le \hat{g}_1(s)\le 1$, where $\gamma_1$ and $\gamma_2$ satisfy $0<\gamma_1<\gamma_2<\infty$.  
Clearly, $[\eta(s)]^2$ can be decomposed into two parts, 
\begin{equation}
\label{eta2decompo}
[\eta(s)]^2=[\eta(s)]^2[\hat{g}_1(s)]^2+[\eta(s)]^2\{1-[\hat{g}_1(s)]^2\}.
\end{equation}
Then, the function $\hat{g}_1$ can be chosen so that the second term satisfies 
\begin{equation}
\label{eta22bound}
[\eta(s)]^2\{1-[\hat{g}_1(s)]^2\}\le{\cal M}_1 s \quad \mbox{for \ } s\ge 0 
\end{equation}
with a small positive parameter ${\cal M}_1$. 
Then, one has 
\begin{equation}
\label{newnumebound} 
\omega_B^{(\Lambda)}(\mathcal{A}_R\mathcal{H}_{\rm p}^{(\Lambda)}(B)
[\eta(\mathcal{H}_{\rm p}^{(\Lambda)}(B))]^2[\hat{g}_1(\mathcal{H}_{\rm p}^{(\Lambda)}(B))]^2\mathcal{A}_R)\le 
\omega_B^{(\Lambda)}(\mathcal{A}_R[\mathcal{H}_{\rm p}^{(\Lambda)}(B)]^{1+\epsilon}\mathcal{A}_R)
\end{equation}
for the left-hand side of (\ref{finalnumebound}). 
Further, we approximate the function $\eta\hat{g}_1$ with a function $\hat{g}\in C_0^\infty(\re)$ such that  
the function $\hat{g}$ satisfies the conditions, $\eta(s)\hat{g}_1(s)\ge \hat{g}(s)\ge 0$ and 
\begin{equation}
\label{hath}
\hat{h}(s):=[\eta(s)\hat{g}_1(s)]^2-[\hat{g}(s)]^2\le\frac{{\cal M}_2}{R^{d-1}} 
\end{equation}
with a small positive parameter ${\cal M}_2$. By using the first condition, we have 
\begin{equation}
\omega_B^{(\Lambda)}(\mathcal{A}_R\hat{g}(\mathcal{H}_{\rm p}^{(\Lambda)}(B))
\mathcal{H}_{\rm p}^{(\Lambda)}(B)\hat{g}(\mathcal{H}_{\rm p}^{(\Lambda)}(B))\mathcal{A}_R)
\le \omega_B^{(\Lambda)}(\mathcal{A}_R[\mathcal{H}_{\rm p}^{(\Lambda)}(B)]^{1+\epsilon}\mathcal{A}_R).
\end{equation}
{from} the above inequality (\ref{newnumebound}). Combining this with (\ref{finalnumebound}), we obtain 
\begin{equation}
\label{newfinalnumebound}
\omega_B^{(\Lambda)}(\mathcal{A}_R\hat{g}(\mathcal{H}_{\rm p}^{(\Lambda)}(B))
\mathcal{H}_{\rm p}^{(\Lambda)}(B)\hat{g}(\mathcal{H}_{\rm p}^{(\Lambda)}(B))\mathcal{A}_R)
\le \frac{1}{R^d}(\mathcal{K}_3+{\mathcal{K}_4}).
\end{equation}

Next, we estimate the quantity $\omega_B^{(\Lambda)}(\mathcal{A}_R
[\mathcal{H}_{\rm p}^{(\Lambda)}(B)]^\epsilon\mathcal{A}_R)$ in the right-hand side of (\ref{finaldenomibound}). 
{From} (\ref{eta2decompo}) and (\ref{eta22bound}), one has 
\begin{eqnarray*}
\omega_B^{(\Lambda)}(\mathcal{A}_R
[\mathcal{H}_{\rm p}^{(\Lambda)}(B)]^\epsilon\mathcal{A}_R)
&\le& \omega_B^{(\Lambda)}(\mathcal{A}_R[\eta(\mathcal{H}_{\rm p}^{(\Lambda)}(B))]^2
[\hat{g}_1(\mathcal{H}_{\rm p}^{(\Lambda)}(B))]^2\mathcal{A}_R)\\
&+&{\cal M}_1\omega_B^{(\Lambda)}(\mathcal{A}_R
\mathcal{H}_{\rm p}^{(\Lambda)}(B)\mathcal{A}_R)\\
&\le& \omega_B^{(\Lambda)}(\mathcal{A}_R[\eta(\mathcal{H}_{\rm p}^{(\Lambda)}(B))]^2
[\hat{g}_1(\mathcal{H}_{\rm p}^{(\Lambda)}(B))]^2\mathcal{A}_R)+\frac{{\cal M}_1{\cal K}_3}{R^d},
\end{eqnarray*}
where we have used the inequality (\ref{omegaBAHAexpvalue}). 
Further, by using (\ref{hath}), the first term in the right-hand side is evaluated as 
\begin{eqnarray}
\omega_B^{(\Lambda)}(\mathcal{A}_R[\eta(\mathcal{H}_{\rm p}^{(\Lambda)}(B))]^2
[\hat{g}_1(\mathcal{H}_{\rm p}^{(\Lambda)}(B))]^2\mathcal{A}_R)
&=&\omega_B^{(\Lambda)}(\mathcal{A}_R[\hat{g}(\mathcal{H}_{\rm p}^{(\Lambda)}(B))]^2\mathcal{A}_R)\nonumber\\
&+&\omega_B^{(\Lambda)}(\mathcal{A}_R\hat{h}(\mathcal{H}_{\rm p}^{(\Lambda)}(B))\mathcal{A}_R)\nonumber \\
&\le&\omega_B^{(\Lambda)}(\mathcal{A}_R[\hat{g}(\mathcal{H}_{\rm p}^{(\Lambda)}(B))]^2\mathcal{A}_R)
+\frac{{\cal M}_2{\cal K}_3}{R^{d-1}}.\nonumber\\
\end{eqnarray}
Combining this with the above inequality, we obtain 
\begin{equation}
\omega_B^{(\Lambda)}(\mathcal{A}_R
[\mathcal{H}_{\rm p}^{(\Lambda)}(B)]^\epsilon\mathcal{A}_R)
\le \omega_B^{(\Lambda)}(\mathcal{A}_R[\hat{g}(\mathcal{H}_{\rm p}^{(\Lambda)}(B))]^2\mathcal{A}_R)
+\frac{{\cal M}_2{\cal K}_3}{R^{d-1}}+\frac{{\cal M}_1{\cal K}_3}{R^d}.
\end{equation}
Therefore, from (\ref{finaldenomibound}), we have 
\begin{equation}
\label{newfinaldenomibound}
\frac{|m_{\rm s}^{(\Lambda)}(B)|^2}{\mathcal{K}_0R^{d-1}\Bigl[1+\mathcal{K}_1
+\mathcal{K}_2|B|R^2\Bigr]^{1/2}}-\frac{{\cal M}_2{\cal K}_3}{R^{d-1}}-\frac{{\cal M}_1{\cal K}_3}{R^d}\le
\omega_B^{(\Lambda)}(\mathcal{A}_R[\hat{g}(\mathcal{H}_{\rm p}^{(\Lambda)}(B))]^2\mathcal{A}_R).
\end{equation}

In order to express the right-hand side of (\ref{newfinaldenomibound}) and the left-hand side of 
(\ref{newfinalnumebound}) in terms of a quasi-local operator, we introduce the time evolution \cite{BR} of local 
operator $a$, 
$$ 
\tau_{t,B}^{(\Lambda)}(a):=\exp[iH_{\rm p}^{(\Lambda)}(B)t]a\exp[-iH_{\rm p}^{(\Lambda)}(B)t], 
$$
and 
$$
\tau_{\ast g,B}^{(\Lambda)}(a):=\int_{-\infty}^{+\infty}dt\; g(t)\tau_{t,B}^{(\Lambda)}(a),
$$
where the function $g$ is the Fourier transform of the function $\hat{g}$. 
Then, the left-hand side of (\ref{newfinalnumebound}) is written \cite{AL,Koma} 
$$
\omega_B^{(\Lambda)}(\mathcal{A}_R\hat{g}(\mathcal{H}_{\rm p}^{(\Lambda)}(B))
\mathcal{H}_{\rm p}^{(\Lambda)}(B)\hat{g}(\mathcal{H}_{\rm p}^{(\Lambda)}(B))\mathcal{A}_R)
=\omega_B^{(\Lambda)}([\tau_{\ast g,B}^{(\Lambda)}(\mathcal{A}_R)]^\ast
[H_{\rm p}^{(\Lambda)}(B),\tau_{\ast g,B}^{(\Lambda)}(\mathcal{A}_R)]).
$$
Similarly, the right-hand side of (\ref{newfinaldenomibound}) is written $$
\omega_B^{(\Lambda)}(\mathcal{A}_R[\hat{g}(\mathcal{H}_{\rm p}^{(\Lambda)}(B))]^2\mathcal{A}_R)
=\omega_B^{(\Lambda)}([\tau_{\ast g,B}^{(\Lambda)}(\mathcal{A}_R)]^\ast\tau_{\ast g,B}^{(\Lambda)}(\mathcal{A}_R)).
$$
Moreover, since the operator $\tau_{\ast g,B}^{(\Lambda)}(a)$ converges to 
$\tau_{* g, 0}(a)$ for $a\in\mathfrak{A}_{\rm loc}$ 
in the double limit \cite{BR} in the definition (\ref{omega0}) of the state $\omega_0$, we have 
$$
\lim_{\Lambda\nearrow\ze^d}\omega_0([\tau_{*g,0}({\cal A}_R)]^\ast[H_{\rm p}^{(\Lambda)}(0),\tau_{*g,0}({\cal A}_R)])
\le \frac{{\cal K}_3+{\cal K}_4}{R^d}
$$
and 
$$
\frac{m_{\rm s}^2}{{\cal K}_0(1+{\cal K}_1)R^{d-1}}
-\frac{{\cal M}_2{\cal K}_3}{R^{d-1}}-\frac{{\cal M}_1{\cal K}_3}{R^d}
\le \omega_0([\tau_{*g,0}({\cal A}_R)]^\ast\tau_{*g,0}({\cal A}_R))
$$
{from} (\ref{newfinalnumebound}) and (\ref{newfinaldenomibound}), where $\omega_0$ is given by (\ref{omega0}).  
Since we can choose the parameters, ${\cal M}_1$ and ${\cal M}_2$, to be small, the latter bound can be written 
$$
\frac{m_{\rm s}^2}{{\cal K}_0'(1+{\cal K}_1)R^{d-1}}
\le \omega_0([\tau_{*g,0}({\cal A}_R)]^\ast\tau_{*g,0}({\cal A}_R))
$$
with some positive constant ${\cal K}_0'$. In consequence, we obtain 
$$
\lim_{\Lambda\nearrow\ze^d}
\frac{\omega_0([\tau_{*g,0}({\cal A}_R)]^\ast[H_{\rm p}^{(\Lambda)}(0),\tau_{*g,0}({\cal A}_R)])}
{\omega_0([\tau_{*g,0}({\cal A}_R)]^\ast\tau_{*g,0}({\cal A}_R))}
\le \frac{{\cal K}_0'(1+{\cal K}_1)({\cal K}_3+{\cal K}_4)}{m_{\rm s}^2R}.
$$
We recall that \cite{BR} the support of the energy cutoff function $\hat{g}$ satisfies 
${\rm supp}\;\hat{g}\subseteq(0,\gamma_2)$
by the definition, and that we can take the parameter $R$ to be any large positive integer.  
The resulting inequality implies that the excited energy spectrum is gapless. 

On the other hand, as we remarked below Theorem~\ref{thm:sponmag}, the state $\omega_0$ exhibits 
the non-vanishing maximum spontaneous magnetization from the assumption of Theorem~\ref{thm:GSgapless} 
that an infinite-volume ground state exhibits a non-vanishing spontaneous magnetization. 
In order to prove the statement of Theorem~\ref{thm:GSgapless}, let us consider the pure-state 
decomposition \cite{BR1} of the state $\omega_0$. Since the state $\omega_0$ shows the maximum spontaneous magnetization, 
almost all the pure states in the decomposition must show the maximum spontaneous magnetization. 
Besides, the state $\omega_0$ shows a gapless excitation as we showed above. 
These imply that there exists an infinite-volume pure ground state which exhibits 
both of a gapless excitation and the maximum spontaneous magnetization.   
Thus, the statement of Theorem~\ref{thm:GSgapless} has been proved.

%%%%%%%%%%%%%%%%%%%%%%%%%%%%%%%%%%%%%%%%%%%%%%%%%
\Section{An Alternative Proof of the Existence of Gapless Excitations} 
\label{sec:AltNambuGoldstone}

In this section, we give an alternative proof of the existence of a gapless excitation 
above an infinite-volume pure ground state. For this purpose, we assume 
the existence of a nonvanishing uniform spectral gap above all the infinite-volume ground states, 
and deduce a contradiction.  
We write $\Delta E$ for the spectral gap. We will estimate the quantity, $\omega_B^{(\Lambda)}(\mathcal{A}_R
[\mathcal{H}_{\rm p}^{(\Lambda)}(B)]^\epsilon\mathcal{A}_R)$, in the right-hand side of (\ref{finaldenomibound}).

To begin with, we decompose the interval $[0,\infty)$ into three parts, $[0,\epsilon')$, $[\epsilon',\Delta E-\epsilon'')$, 
and $[\Delta E-\epsilon'',\infty)$, where the two positive parameters, $\epsilon'$ and $\epsilon''$, satisfy 
$0<\epsilon'<\Delta E-\epsilon''$, and write $P[0,\epsilon'), P[\epsilon',\Delta E-\epsilon'')$, and 
$P[\Delta E-\epsilon'',\infty)$ for the corresponding spectral projections 
for the Hamiltonian $\mathcal{H}_{\rm p}^{(\Lambda)}(B)$. 
Note that 
\begin{eqnarray*}
\omega_B^{(\Lambda)}(\mathcal{A}_R
[\mathcal{H}_{\rm p}^{(\Lambda)}(B)]^\epsilon\mathcal{A}_R)&=&
\omega_B^{(\Lambda)}(\mathcal{A}_RP[0,\epsilon')
[\mathcal{H}_{\rm p}^{(\Lambda)}(B)]^\epsilon\mathcal{A}_R)\\
&+&\omega_B^{(\Lambda)}(\mathcal{A}_R[1-P[0,\epsilon')]
[\mathcal{H}_{\rm p}^{(\Lambda)}(B)]^\epsilon\mathcal{A}_R).
\end{eqnarray*}
The first term in the right-hand side can be estimated as 
$$
\omega_B^{(\Lambda)}(\mathcal{A}_RP[0,\epsilon')
[\mathcal{H}_{\rm p}^{(\Lambda)}(B)]^\epsilon\mathcal{A}_R)\le 
(\epsilon')^\epsilon\times \omega_B^{(\Lambda)}(\mathcal{A}_RP[0,\epsilon')
\mathcal{A}_R)\le{\rm Const.}(\epsilon')^\epsilon. 
$$
Therefore, this is vanishing in the limit $\epsilon'\rightarrow 0$ after taking the double limit, 
$B\searrow 0$ and $\Lambda\nearrow\ze^d$. 
The second term can be further decomposed into two parts,  
\begin{eqnarray}
\label{decomomegaARP0AR}
\omega_B^{(\Lambda)}(\mathcal{A}_R[1-P[0,\epsilon')]
[\mathcal{H}_{\rm p}^{(\Lambda)}(B)]^\epsilon\mathcal{A}_R)&=&
\omega_B^{(\Lambda)}(\mathcal{A}_RP[\epsilon',\Delta E-\epsilon'')
[\mathcal{H}_{\rm p}^{(\Lambda)}(B)]^\epsilon\mathcal{A}_R)\nonumber \\
&+&\omega_B^{(\Lambda)}(\mathcal{A}_RP[\Delta E-\epsilon'',\infty)
[\mathcal{H}_{\rm p}^{(\Lambda)}(B)]^\epsilon\mathcal{A}_R).
\end{eqnarray}
The first and second terms in the right-hand side are evaluated as 
\begin{equation}
\label{gapAssumption}
\omega_B^{(\Lambda)}(\mathcal{A}_RP[\epsilon',\Delta E-\epsilon'')
[\mathcal{H}_{\rm p}^{(\Lambda)}(B)]^\epsilon\mathcal{A}_R)
\le (\Delta E)^\epsilon\times \omega_B^{(\Lambda)}(\mathcal{A}_RP[\epsilon',\Delta E-\epsilon'')
\mathcal{A}_R),
\end{equation}
and 
\begin{eqnarray}
\label{decomomegaARPEpAR}
\omega_B^{(\Lambda)}(\mathcal{A}_RP[\Delta E-\epsilon'',\infty)
[\mathcal{H}_{\rm p}^{(\Lambda)}(B)]^\epsilon\mathcal{A}_R)&\le& 
\frac{1}{(\Delta E-\epsilon'')^{1-\epsilon}}\times \omega_B^{(\Lambda)}(\mathcal{A}_R\mathcal{H}_{\rm p}^{(\Lambda)}(B)
\mathcal{A}_R) \nonumber\\
&\le& \frac{{\cal K}_3}{(\Delta E-\epsilon'')^{1-\epsilon}R^d}, 
\end{eqnarray}
where we have used the bound (\ref{omegaBAHAexpvalue}). 
The quantity in the right-hand side of (\ref{gapAssumption}) is vanishing as 
$$
\lim_{B\searrow 0}\lim_{\Lambda \nearrow\ze^d}\omega_B^{(\Lambda)}(\mathcal{A}_RP[\epsilon',\Delta E-\epsilon'')
\mathcal{A}_R)=0 
$$
{from} the assumption on the spectral gap. Actually, one can prove this statement in the same way as in 
Sec.~\ref{ProofTheoremGSgapless}. 
 
In consequence, the nonvanishing contribution is only the quantity in (\ref{decomomegaARPEpAR}). Therefore, we have   
\begin{equation}
\label{limitsomegaARHepsAR}
\lim_{B\searrow 0}\lim_{\Lambda\nearrow \ze^d}\omega_B^{(\Lambda)}(\mathcal{A}_R
[\mathcal{H}_{\rm p}^{(\Lambda)}(B)]^\epsilon\mathcal{A}_R)
\le\frac{{\cal K}_3}{(\Delta E-\epsilon'')^{1-\epsilon}R^d}. 
\end{equation}
Combining this and (\ref{finaldenomibound}), we obtain 
\begin{eqnarray*}
|m_{\rm s}|^2&=&\lim_{B\searrow 0}\lim_{\Lambda\nearrow \ze^d}|m_{\rm s}^{(\Lambda)}(B)|^2\\
&\le& \mathcal{K}_0R^{d-1}(1+{\cal K}_1)\lim_{B\searrow 0}\lim_{\Lambda\nearrow \ze^d}
\omega_B^{(\Lambda)}(\mathcal{A}_R[\mathcal{H}_{\rm p}^{(\Lambda)}(B)]^\epsilon\mathcal{A}_R)
\le\frac{\mathcal{K}_0(1+{\cal K}_1){\cal K}_3}{(\Delta E-\epsilon'')^{1-\epsilon}R}.
\end{eqnarray*}
Since the spontaneous magnetization $m_{\rm s}$ is strictly positive, this is a contradiction for a sufficiently large $R$. 

This proof is much simpler than that in the preceding section. 
The proof of the preceding Sec.~\ref{ProofTheoremGSgapless} guarantees that the contribution that the pure infinite-volume 
ground states show both of the maximum spontaneous magnetization and a gapless excitation  
is nonvanishing in the sense of the measure of the pure state decomposition for the infinite-volume ground state $\omega_0$. 
However, in the proof of the present section, the set of the pure infinite-volume 
ground states showing both of the maximum spontaneous magnetization and a gapless excitation 
is allowed to consist of a single point in the pure state decomposition. 
For example, such a single point may be given by the limit point of the sequence of the ground states, $\omega_n$, 
$n=1,2,\ldots$, whose excitation energy gap above $\omega_n$ is given by $\Delta E_n>0$ which satisfies 
$\Delta E_n\rightarrow 0$ as $n\rightarrow \infty$.

%%%%%%%%%%%%%%%%%%%%%%%%%%%%%%%%%%%%%%%%%%%%%%%%%%%%%%%%%
\Section{Slowly-Decaying Transverse Correlations}
\label{sec:TransCorr}

The effect of continuous symmetry breaking is reflected in the emergence of the slowly-decaying transverse correlations. 
In this section, we will prove this statement for the present model at both of zero and non-zero temperatures. 
Namely, we will give proofs of Theorems~\ref{theorem:slowdecayZero} and \ref{theorem:slowdecayNonzero}. 

%%%%%%%%%%%%%%%%%%%%%%%%%%%%%%%%%%
\subsection{Non-zero temperatures}

Consider first the transverse correlation function, $\langle S_x^{(2)}S_y^{(2)}\rangle_{B,\beta}^{(\Lambda)}$, 
for strictly positive temperatures $\beta^{-1}>0$ in dimensions $d\ge 3$. 
Here, the expectation $\langle\cdots\rangle_{B,\beta}^{(\Lambda)}$ 
is given by (\ref{thermalstate}).  
Since continuous symmetry breaking does not occur for strictly positive temperatures in dimensions $d\le 2$, 
we will consider only the case with $d\ge 3$ for strictly positive temperatures $\beta^{-1}>0$. 

For the purpose of the present section, we first prove Lemma~\ref{lemma:vanishTransCorr} below 
which states that the transverse correlation is vanishing in the large distance limit $|x-y|\rightarrow\infty$, 
and then we estimate the speed of the decay of the transverse correlations by using 
Bogoliubov inequality (\ref{Bogoineq}) below \cite{DLS}. 

\begin{lemma}
\label{lemma:vanishTransCorr}
Let $\rho_0$ be the state given by (\ref{def:rho0}) with (\ref{rhoB}). Then, we have
\begin{equation}
\label{vanishTransCorr}
\lim_{|x-y|\rightarrow \infty}\rho_0(S_x^{(2)}S_y^{(2)})= 0.
\end{equation}
\end{lemma}

In order to prove Lemma~\ref{lemma:vanishTransCorr}, we use the following lemma \cite{FSS,DLS}:  
\begin{lemma} 
\label{lemma:LROtranscorr}
\begin{equation}
\label{LROtranscorr}
\lim_{|x-y|\rightarrow \infty}\rho_0(S_x^{(2)}S_y^{(2)})
= v_0(-1)^{x^{(1)}-y^{(1)}+\cdots+x^{(d)}-y^{(d)}}
\end{equation}
with some constant $v_0$. 
\end{lemma}

The proof is given in Appendix~\ref{Appendix:prooflem:LROtranscorr}.
\bigskip 

\begin{proof}{Proof of Lemma~\ref{lemma:vanishTransCorr}}
In order to prove the statement of Lemma~\ref{lemma:vanishTransCorr} by relying on (\ref{LROtranscorr}), 
we use the ergodic decomposition \cite{BR1} of the translationally invariant 
equilibrium state $\rho_0$, 
$$
\rho_0(\cdots)=\int d\nu(\lambda)\rho_{0,\lambda}(\cdots), 
$$
where $\rho_{0,\lambda}$ is an ergodic state, and $\nu$ is the probability measure on the set of the extremal points. 

Consider first the magnetization in the second direction which is given by 
$$
\rho_0(\mathcal{A}_R)=\int d\nu(\lambda)\rho_{0,\lambda}(\mathcal{A}_R), 
$$
where the operator $\mathcal{A}_R$ is given by (\ref{AR}). Since the state $\rho_0$ yields 
the maximum magnetization in the third direction as proved in Sec.~\ref{Sec:SponMagFiniteT}, 
almost all the states $\rho_{0,\lambda}$ of the integrand also yield the maximum magnetization in the same direction. 
This implies that almost all the transverse magnetizations $\rho_{0,\lambda}(\mathcal{A}_R)$ are vanishing. 
Actually, this can be shown as follows: 
If the transverse magnetization is non-vanishing, then the magnetization vector may be written 
${\bf m}:=(0,m^{(2)},m^{(3)})$ with $m^{(2)}\ne 0$. Since the third component $m^{(3)}$ takes the value $m_{\rm s}$ of 
the maximum spontaneous magnetization, the magnitude of the magnetization vector ${\bf m}$ becomes 
$\sqrt{(m^{(2)})^2+(m_{\rm s})^2}>|m_{\rm s}|$. This contradicts with the fact that $m_{\rm s}$ takes 
the supremum over all the magnetizations. 

Next, consider the long-range order which is given by 
$$
\rho_0(\mathcal{A}_R^2)=\int d\nu(\lambda)\rho_{0,\lambda}(\mathcal{A}_R^2). 
$$
Since the expectation value $\rho_{0,\lambda}(\mathcal{A}_R)$ is vanishing,  
the ergodicity of the states $\rho_{0,\lambda}$ yields that  
almost all the expectation values $\rho_{0,\lambda}(\mathcal{A}_R^2)$ must be vanishing in 
the limit $R\nearrow \infty$.  
Combining this with the above result (\ref{LROtranscorr}), we obtain 
$$
v_0=\lim_{R\nearrow\infty}\rho_0(\mathcal{A}_R^2)=0. 
$$
Consequently, we have (\ref{vanishTransCorr}).   
\end{proof}

Next, in order to estimate the speed of the decay of the transverse correlations, 
we use Bogoliubov inequality \cite{DLS},  
\begin{equation}
\label{Bogoineq}
|\langle[C,A]\rangle_{B,\beta}^{(\Lambda)}|^2\le \frac{\beta}{2}
\langle[C,[H_{\rm p}^{(\Lambda)}(B),C^\ast]]\rangle_{B,\beta}^{(\Lambda)}\langle\{A,A^\ast\}\rangle_{B,\beta}^{(\Lambda)}, 
\end{equation}
for operators $A$ and $C$. Set $A=\mathcal{A}_R$ and 
$$
C=\sum_x S_x^{(1)} f_x, 
$$
where $f_x$ is given by (\ref{fx}). Then, one has 
$$
|m_{{\rm s},\beta}^{(\Lambda)}(B)|^2\le \beta \mathcal{K}_5R^{d-2}
[1+\mathcal{K}_6|B|R^2]\langle \mathcal{A}_R^2\rangle_{B,\beta}^{(\Lambda)}
$$
in the same way as the derivation of the inequality (\ref{denomibound}), where we have written 
$$
m_{{\rm s},\beta}^{(\Lambda)}(B):=\frac{1}{|\Omega_R|}\langle O^{(\Omega_R)}\rangle_{B,\beta}^{(\Lambda)}, 
$$
and $\mathcal{K}_5$ and $\mathcal{K}_6$ are positive constants. 
Further, in the same way as the derivation of (\ref{mslimitzero}), we obtain 
$$
m_{{\rm s},\beta}=\frac{1}{|\Omega_R|}\rho_0(O^{(\Omega_R})=
\lim_{B\searrow 0}\lim_{\Lambda\nearrow\ze^d}m_{{\rm s},\beta}^{(\Lambda)}(B),
$$
where the sequences in the double limit are taken to be the same as in those for the state $\rho_0$ of (\ref{def:rho0}). 
Therefore, in the same double limit, one has 
\begin{equation}
\label{SponMagBoundFiniteT}
|m_{{\rm s},\beta}|^2\le \mathcal{K}_5{\beta}R^{d-2}\rho_0(\mathcal{A}_R^2). 
\end{equation}
When the spontaneous magnetization $m_{{\rm s},\beta}$ in the left-hand side is non-vanishing, 
this bound rules out the possibility of rapid decay $o(|x-y|^{-(d-2)})$ for the correlation. 
Thus, the transverse correlation $\rho_0(S_x^{(2)}S_y^{(2)})$ exhibits slow decay. 

We have proved $\rho_0(\mathcal{A}_R^2)\rightarrow 0$ as $R\nearrow \infty$, 
irrespective of the dimension $d$ of the lattice, 
only {from} the argument of the maximum spontaneous magnetization. Therefore, we can obtain a slightly stronger result than 
Hohenberg-Mermin-Wagner theorem \cite{Hohenberg,MW} from the inequality (\ref{SponMagBoundFiniteT}) as:  

\begin{coro}
In dimensions $d\le 2$ for strictly positive temperatures $\beta^{-1}>0$, the maximum spontaneous magnetization is vanishing 
in translationally invariant thermal equilibrium states. 
\end{coro}

%%%%%%%%%%%%%%%%%%%%%%%%%
\subsection{Zero temperature}

Next, consider the case of the ground state in dimensions $d\ge 2$. 
In the zero temperature limit $\beta\nearrow\infty$, the thermal equilibrium state $
\langle\cdots\rangle_{B,\beta}^{(\Lambda)}$ of (\ref{thermalstate}) coincides with 
the ground state $\omega_B^{(\Lambda)}$ of (\ref{Grstate}). 

In order to clarify the difference between the zero and non-zero temperature cases, 
we recall first the quantity, 
$$
g_p^{(\Lambda)}(B,\beta):=\frac{1}{2}\left[\langle \hat{S}_p^{(2)}\hat{S}_{-p}^{(2)}\rangle_{B,\beta}^{(\Lambda)}+
\langle \hat{S}_{-p}^{(2)}\hat{S}_{p}^{(2)}\rangle_{B,\beta}^{(\Lambda)}\right],
$$
which is given by (\ref{gp}) in Appendix~\ref{Appendix:prooflem:LROtranscorr}. 
Here, 
$$
\hat{S}_p^{(2)}:=|\Lambda|^{-1/2}\sum_{x\in\Lambda}e^{-ipx}S_x^{(2)}
$$
with the wavevector $p$. 
{From} the inequalities (\ref{bpbound}) and (\ref{gpbound}) in Appendix~\ref{Appendix:prooflem:LROtranscorr}, 
one has \cite{JNFP} 
\begin{equation}
g_p^{(\Lambda)}(B)\le \frac{1}{2}\sqrt{\frac{c_p^{(\Lambda)}(B)}{2\mathcal{E}_p'}}
\end{equation}
in the zero temperature limit $\beta\nearrow\infty$, where we have written 
$$
g_p^{(\Lambda)}(B):=\lim_{\beta\nearrow\infty}g_p^{(\Lambda)}(B,\beta)
$$
and 
$$
c_p^{(\Lambda)}(B):=\lim_{\beta\nearrow\infty}c_p^{(\Lambda)}(B,\beta). 
$$
Here, the quantity $c_p^{(\Lambda)}(B,\beta)$ in the right-hand side 
is given by (\ref{cp}) in Appendix~\ref{Appendix:prooflem:LROtranscorr}. 
Since one has 
$$
c_p^{(\Lambda)}(B)\le 4S^2\mathcal{E}_p+{\rm Const.}|B| 
$$
{from} (\ref{cpbound}) in Appendix~\ref{Appendix:prooflem:LROtranscorr}, 
the following bound is valid: 
$$
g_p^{(\Lambda)}(B)\le
\sqrt{\frac{S^2\mathcal{E}_p+{\rm Const.}|B|}{2\mathcal{E}_p'}}.
$$
This implies 
\begin{equation}
\label{IRBomega0}
\omega_B^{(\Lambda)}(\hat{S}_p^{(2)}\hat{S}_{-p}^{(2)})\le \sqrt{\frac{2S^2\mathcal{E}_p+{\rm Const.}|B|}{\mathcal{E}_p'}}
\end{equation}
{from} the definition of the function $g_p^{(\Lambda)}(B,\beta)$ in the same way as 
in Appendix~\ref{Appendix:prooflem:LROtranscorr}. 
In the case of non-zero temperatures, the corresponding bound is given by (\ref{SpCorrbound}). 
In the present case of zero temperature, the first term in the right-hand side of (\ref{SpCorrbound}) 
is absent. Therefore, the right-hand side of (\ref{IRBomega0}) is integrable with respect to the wavevector $p$ in two or 
higher dimensions except for the singularity at $p=(\pi,\ldots,\pi)$. 
The same argument as in the case for strictly positive temperatures yields 
\begin{equation}
\lim_{|x-y|\rightarrow \infty}\omega_0(S_x^{(2)}S_y^{(2)})= 0,
\end{equation}
where the infinite-volume ground state $\omega_0$ is given by (\ref{omega0}). 
Thus, the transverse correlation decays in the large distance limit. 

Further, in order to estimate the speed of the decay, we recall the inequality (\ref{denomibound}),  
which holds also for $\epsilon=0$. As a result, the spontaneous magnetization $m_{\rm s}$ satisfies  
\begin{equation}
|m_{\rm s}|^2\le\mathcal{K}_0R^{d-1}\omega_0(\mathcal{A}_R^2).
\end{equation}
When the spontaneous magnetization $m_{\rm s}$ in the left-hand side is non-vanishing, 
this bound rules out the possibility of rapid decay $o(|x-y|^{-(d-1)})$ for the correlation.
Remarkably, the exponent $(d-1)$ is different from $(d-2)$ in the case for strictly positive temperatures. 

Further, the quantity $\omega_0(\mathcal{A}_R^2)$ is vanishing in the limit $R\nearrow\infty$ 
because the same argument about the maximum spontaneous magnetization 
as in the case for strictly positive temperatures holds.  
In consequence, we obtain a result which is slightly stronger than Shastry theorem \cite{Shastry} 
in one dimension at zero temperature as: 

\begin{coro}
In one dimension, the maximum spontaneous magnetization is vanishing 
in the sector of the infinite-volume ground states. 
\end{coro}
\medskip

\noindent
Remark: The method by Shastry was applied to a one-dimensional spin-orbital model \cite{TTI}. 

%%%%%%%%%%%%%%%%%%%%%%%%%%%%%%%%%%%%%%%%%%%%%%%%%
\appendix
 
\Section{Proof of Lemma~\ref{lem:BogolyIneq}}
\label{App::BogolyIneq}

In this appendix, we give a proof of the inequality (\ref{BogolyIneq}) in Lemma~\ref{lem:BogolyIneq}.  

By the cyclic property of the trace and the definition (\ref{Grstate}) of the ground state $\omega_B^{(\Lambda)}$, 
the quantity in the left-hand side of (\ref{BogolyIneq}) can be written as  
\begin{eqnarray*}
\omega_B^{(\Lambda)}([C,A])&=&\omega_B^{(\Lambda)}(CA)-\omega_B^{(\Lambda)}(AC)\\
&=&\omega_B^{(\Lambda)}(CP_{\rm ex}^{(\Lambda)}(B)A)-\omega_B^{(\Lambda)}(AP_{\rm ex}^{(\Lambda)}(B)C). 
\end{eqnarray*}
The first term in the right-hand side can be estimated as 
\begin{eqnarray*}
|\omega_B^{(\Lambda)}(CP_{\rm ex}^{(\Lambda)}(B)A)|^2
&=&|\omega_B^{(\Lambda)}(C[\mathcal{H}_{\rm p}^{(\Lambda)}(B)]^{-\epsilon/2}P_{\rm ex}^{(\Lambda)}(B)
[\mathcal{H}_{\rm p}^{(\Lambda)}(B)]^{\epsilon/2}A)|^2\\
&\le& 
\omega_B^{(\Lambda)}(CP_{\rm ex}^{(\Lambda)}(B)[\mathcal{H}_{\rm p}^{(\Lambda)}(B)]^{-\epsilon}C^\ast)
\; \omega_B^{(\Lambda)}(A^\ast P_{\rm ex}^{(\Lambda)}(B)[\mathcal{H}_{\rm p}^{(\Lambda)}(B)]^\epsilon A),
\end{eqnarray*}
where we have used the positivity of 
$\mathcal{H}_{\rm p}^{(\Lambda)}(B)=H_{\rm p}^{(\Lambda)}(B)-E_0^{(\Lambda)}(B)$ and Schwarz inequality, 
and $\epsilon$ is a small positive parameter. 
Since the second term can be handled in the same way, we have 
\begin{eqnarray*}
|\omega_B^{(\Lambda)}([C,A])|
&\le&
\left[\omega_B^{(\Lambda)}(CP_{\rm ex}^{(\Lambda)}(B)[\mathcal{H}_{\rm p}^{(\Lambda)}(B)]^{-\epsilon}C^\ast)\right]^{1/2}
\left[\omega_B^{(\Lambda)}(A^\ast P_{\rm ex}^{(\Lambda)}(B)[\mathcal{H}_{\rm p}^{(\Lambda)}(B)]^\epsilon A)\right]^{1/2}\\
&+&
\left[\omega_B^{(\Lambda)}(C^\ast P_{\rm ex}^{(\Lambda)}(B)[\mathcal{H}_{\rm p}^{(\Lambda)}(B)]^{-\epsilon}C)\right]^{1/2}
\left[\omega_B^{(\Lambda)}(AP_{\rm ex}^{(\Lambda)}(B)[\mathcal{H}_{\rm p}^{(\Lambda)}(B)]^\epsilon A^\ast)\right]^{1/2}.
\end{eqnarray*}
Further, by using the inequality $2ab\le(a^2+b^2)$ for $a,b>0$, we obtain 
\begin{eqnarray}
\label{omega[C,A]2bound}
|\omega_B^{(\Lambda)}([C,A])|^2
&\le& \Bigl\{\omega_B^{(\Lambda)}(CP_{\rm ex}^{(\Lambda)}(B)[\mathcal{H}_{\rm p}^{(\Lambda)}(B)]^{-\epsilon}C^\ast)
+\omega_B^{(\Lambda)}(C^\ast P_{\rm ex}^{(\Lambda)}(B)[\mathcal{H}_{\rm p}^{(\Lambda)}(B)]^{-\epsilon}C)\Bigr\}
\nonumber \\
&\times&
\Bigl\{\omega_B^{(\Lambda)}(AP_{\rm ex}^{(\Lambda)}(B)[\mathcal{H}_{\rm p}^{(\Lambda)}(B)]^\epsilon A^\ast)
+\omega_B^{(\Lambda)}(A^\ast P_{\rm ex}^{(\Lambda)}(B)[\mathcal{H}_{\rm p}^{(\Lambda)}(B)]^\epsilon A)\Bigr\}.\nonumber\\
& &
\end{eqnarray}

Let us consider the quantities in the right-hand side of (\ref{omega[C,A]2bound}) 
which include the operator $C$. By using Schwarz inequality, we have 
\begin{eqnarray*}
& &\omega_B^{(\Lambda)}(CP_{\rm ex}^{(\Lambda)}(B)[\mathcal{H}_{\rm p}^{(\Lambda)}(B)]^{-\epsilon}C^\ast)\\
&=&\omega_B^{(\Lambda)}(CP_{\rm ex}^{(\Lambda)}(B)[\mathcal{H}_{\rm p}^{(\Lambda)}(B)]^{-1/2}
[\mathcal{H}_{\rm p}^{(\Lambda)}(B)]^{1/2-\epsilon}C^\ast)\\
&\le&\sqrt{\omega_B^{(\Lambda)}(CP_{\rm ex}^{(\Lambda)}(B)[\mathcal{H}_{\rm p}^{(\Lambda)}(B)]^{-1}C^\ast)} 
\sqrt{\omega_B^{(\Lambda)}(C[\mathcal{H}_{\rm p}^{(\Lambda)}(B)]^{1-2\epsilon}C^\ast)}.
\end{eqnarray*}
Therefore, 
\begin{eqnarray*}
& &\omega_B^{(\Lambda)}(CP_{\rm ex}^{(\Lambda)}(B)[\mathcal{H}_{\rm p}^{(\Lambda)}(B)]^{-\epsilon}C^\ast)
+\omega_B^{(\Lambda)}(C^\ast P_{\rm ex}^{(\Lambda)}(B)[\mathcal{H}_{\rm p}^{(\Lambda)}(B)]^{-\epsilon}C)\\
&\le &\sqrt{\omega_B^{(\Lambda)}(CP_{\rm ex}^{(\Lambda)}(B)[\mathcal{H}_{\rm p}^{(\Lambda)}(B)]^{-1}C^\ast)} 
\sqrt{\omega_B^{(\Lambda)}(C[\mathcal{H}_{\rm p}^{(\Lambda)}(B)]^{1-2\epsilon}C^\ast)}\\
&+&\sqrt{\omega_B^{(\Lambda)}(C^\ast P_{\rm ex}^{(\Lambda)}(B)[\mathcal{H}_{\rm p}^{(\Lambda)}(B)]^{-1}C)} 
\sqrt{\omega_B^{(\Lambda)}(C^\ast[\mathcal{H}_{\rm p}^{(\Lambda)}(B)]^{1-2\epsilon}C)}. 
\end{eqnarray*}
Further, by using $2ab\le(a^2+b^2)$ for $a,b>0$, we obtain
\begin{eqnarray*}
& &\omega_B^{(\Lambda)}(CP_{\rm ex}^{(\Lambda)}(B)[\mathcal{H}_{\rm p}^{(\Lambda)}(B)]^{-\epsilon}C^\ast)
+\omega_B^{(\Lambda)}(C^\ast P_{\rm ex}^{(\Lambda)}(B)[\mathcal{H}_{\rm p}^{(\Lambda)}(B)]^{-\epsilon}C)\\
&\le&\sqrt{\tilde{D}_B^{(\Lambda)}(C)}
\Bigl\{\omega_B^{(\Lambda)}(C[\mathcal{H}_{\rm p}^{(\Lambda)}(B)]^{1-2\epsilon}C^\ast)
+\omega_B^{(\Lambda)}(C^\ast[\mathcal{H}_{\rm p}^{(\Lambda)}(B)]^{1-2\epsilon}C)\Bigr\}^{1/2}\\
&\le&\sqrt{\tilde{D}_B^{(\Lambda)}(C)}
\Bigl\{\kappa(\epsilon)\;\omega_B^{(\Lambda)}(\{C,C^\ast\})+\omega_B^{(\Lambda)}([[C^\ast,H_{\rm p}^{(\Lambda)}(B)],C])\Bigr\}^{1/2},
\end{eqnarray*}
where $\tilde{D}_B^{(\Lambda)}(C)$ is given by (\ref{tildeD}), and for deriving the last inequality, we have used 
the following inequality:  
$$
t^{1-2\epsilon}\le \kappa(\epsilon)+t \quad \mbox{for \ } t\ge 0.
$$
Here, $\kappa(\epsilon)$ is a positive function of the parameter $\epsilon>0$ such that $\kappa(\epsilon)\rightarrow 0$ as 
$\epsilon\rightarrow 0$. 
Substituting this bound into the right-hand side of (\ref{omega[C,A]2bound}), we obtain the desired bound (\ref{BogolyIneq}). 

%%%%%%%%%%%%%%%%%%%%%%%%%%%%%%%
\Section{Proof of the Inequality (\ref{E0bound})}
\label{RP}

We give a proof of the bound (\ref{E0bound}), following Kennedy, Lieb and Shastry \cite{KLS}. 
To begin with, we remark the following: 
In their proof of the upper bound for susceptibility (\ref{chibound}) which can be derived from (\ref{E0bound}), 
they used the uniqueness \cite{LM} of the ground state of 
the finite-volume Hamiltonian. However, in a situation where a magnetic field is applied, 
the uniqueness of the ground state does not necessarily hold \cite{LM}. 
Therefore, we do not assume the uniqueness of the ground state for deriving (\ref{chibound}) 
in Appendix~\ref{appen:chibound}. 
 
To begin with, we recall the Hamiltonian (\ref{HamBf}) as 
\begin{eqnarray}
H_{\rm p}^{(\Lambda)}(B,f)&:=&\sum_{\{x,y\}\subset\Lambda:|x-y|=1}\left[S_x^{(2)}S_y^{(2)}+S_x^{(3)}S_y^{(3)}\right]
-B\sum_{x\in\Lambda} (-1)^{x^{(1)}+\cdots+x^{(d)}}S_x^{(3)}\nonumber \\
&+&\frac{1}{2}\sum_{\{x,y\}\subset\Lambda:|x-y|=1}\left[(S_x^{(1)}+S_y^{(1)}+f_x+f_y)^2-(S_x^{(1)})^2-(S_y^{(1)})^2\right].
\end{eqnarray}
By using the unitary transformation which is rotation by $\pi$ about the $2$ axis in the spin space 
at site $x$ for all the sites $x$ with odd $(x^{(1)}+\cdots+x^{(d)})$, the Hamiltonian is transformed as   
\begin{eqnarray}
\label{Hamrecall}
\tilde{H}_{\rm p}^{(\Lambda)}(B,\tilde{f})&:=&\sum_{\{x,y\}\subset\Lambda:|x-y|=1}
\left[S_x^{(2)}S_y^{(2)}-S_x^{(3)}S_y^{(3)}\right]
-B\sum_{x\in\Lambda} S_x^{(3)}\nonumber \\
&+&\frac{1}{2}\sum_{\{x,y\}\subset\Lambda:|x-y|=1}\left[(S_x^{(1)}-S_y^{(1)}-\tilde{f}_x
+\tilde{f}_y)^2-(S_x^{(1)})^2-(S_y^{(1)})^2\right],
\end{eqnarray}
where
$$
\tilde{f}_x:=-(-1)^{x^{(1)}+\cdots+x^{(d)}}f_x.
$$
We write $\tilde{E}_0^{(\Lambda)}(B,\tilde{f})$ for the energy of the ground state 
of the Hamiltonian $\tilde{H}_{\rm p}^{(\Lambda)}(B,\tilde{f})$. 
Our aim in this appendix is to prove the bound, 
\begin{equation}
\label{energytildefbound}
\tilde{E}_0^{(\Lambda)}(B,\tilde{f})\ge\tilde{E}_0^{(\Lambda)}(B,0),
\end{equation}
for any real-valued function $\tilde{f}$ on the lattice $\Lambda$. Clearly, from the expression of (\ref{Hamrecall}), 
this is equivalent to showing 
that the energy $\tilde{E}_0^{(\Lambda)}(B,\tilde{f})$ takes its minimum value when $\tilde{f}$ is a constant.  
We assume that $\tilde{E}_0^{(\Lambda)}(B,\cdots)$ takes its minimum value for a real-valued function $\overline{f}$ 
which has the least number of bonds $\{x,y\}$ with $\overline{f}_x\not=\overline{f}_y$
in a set of the configurations $f$ which minimize the energy, and we deduce a contradiction 
if the number of those bonds is not equal to zero. 

Let $\{x_0,y_0\}$ be a bond satisfying $\overline{f}_{x_0}\not=\overline{f}_{y_0}$ for the above function $\overline{f}$. 
We draw a plane through the midpoint of the bond $\{x_0,y_0\}$ and perpendicular to the bond. 
Further, we draw a second plane which is parallel to the first one but shifted by $L$, remembering that 
$2L$ is the sidelength of the lattice $\Lambda$, and that the periodic boundary conditions are imposed.
Clearly, these two planes, which will be denoted collectively by $\Pi$, divide the lattice $\Lambda$ into two parts, 
$\Lambda^{\rm L}$ and $\Lambda^{\rm R}$, which will be referred to as the left and right halves, respectively.  

In the following, we will use the usual real, orthonormal basis of $S^{(3)}$ eigenstates. 
We denote by $\Psi_\alpha^{\rm L}$ and $\Psi_\beta^{\rm R}$ the basis vectors which are associated with 
the left and right half Hilbert spaces, respectively. The basis for the full Hilbert space is given by 
$\Psi_\alpha^{\rm L}\otimes\Psi_\beta^{\rm R}$. 
We set $T_x^{(2)}=iS_x^{(2)}$ for $x\in\Lambda$. Then, the Hamiltonian $\tilde{H}_{\rm p}^{(\Lambda)}(B,\tilde{f})$ 
of (\ref{Hamrecall}) can be written as 
\begin{eqnarray}
\label{tildeHT}
\tilde{H}_{\rm p}^{(\Lambda)}(B,\tilde{f})
&:=&-\sum_{\{x,y\}\subset\Lambda:|x-y|=1}\left[T_x^{(2)}T_y^{(2)}+S_x^{(3)}S_y^{(3)}
+(S_x^{(1)}-\tilde{f}_x)(S_y^{(1)}-\tilde{f}_y)\right]\nonumber \\
& &+ \sum_{x\in\Lambda} \left[d \tilde{f}_x(\tilde{f}_x-2S_x^{(1)})-BS_x^{(3)}\right].
\end{eqnarray}
Therefore, this Hamiltonian has real matrix elements in this basis, and a ground state $\Psi$ of the Hamiltonian 
can be written as 
\begin{equation}
\label{PsiGS}
\Psi=\sum_{\alpha,\beta}\mathcal{C}_{\alpha,\beta}\Psi_\alpha^{\rm L}\otimes \Psi_\beta^{\rm R}
\end{equation}
in terms of real numbers $\mathcal{C}_{\alpha,\beta}$. 

Clearly, there are three types of bonds: Bonds with both endpoints in the left half $\Lambda^{\rm L}$ will 
be referred to as ``left" bonds. The ``right" bonds are defined in the same way. 
Bonds with one endpoint in the left half $\Lambda^{\rm L}$ and the other in the right half $\Lambda^{\rm R}$ 
will be referred to as ``crossing". 

For the Hamiltonian (\ref{tildeHT}), we write $H$ for short.  
We define by $H^{\rm L}$ the sum of all the terms in the Hamiltonian $H$ labeled by left bonds and sites in the left half 
$\Lambda^{\rm L}$. Similarly, $H^{\rm R}$ is defined. 
We denote the bonds crossing the planes $\Pi$ by $\{x_i,y_i\}$ with $x_i$ in the left half $\Lambda^{\rm L}$ and 
$y_i$ in the right half $\Lambda^{\rm R}$. Then, one has 
$$
H=H^{\rm L}+H^{\rm R}-\sum_i\left[T_{x_i}^{(2)}T_{y_i}^{(2)}+S_{x_i}^{(3)}S_{y_i}^{(3)}
+(S_{x_i}^{(1)}-\tilde{f}_{x_i})(S_{y_i}^{(1)}-\tilde{f}_{y_i})\right].
$$
We write 
$$
H_{\alpha,\gamma}^{\rm L}:=\langle \Psi_\alpha^{\rm L},H^{\rm L}\Psi_\gamma^{\rm L}\rangle
$$
and 
$$
X_{\alpha,\gamma}^{{\rm L},i}:=\langle \Psi_\alpha^{\rm L},(S_{x_i}^{(1)}-\tilde{f}_{x_i})\Psi_\gamma^{\rm L}\rangle,
\quad Y_{\alpha,\gamma}^{{\rm L},i}:=\langle\Psi_\alpha^{\rm L},T_{x_i}^{(2)}\Psi_\gamma^{\rm L}\rangle,
\quad Z_{\alpha,\gamma}^{{\rm L},i}:=\langle\Psi_\alpha^{\rm L},S_{x_i}^{(3)}\Psi_\gamma^{{\rm L},i}\rangle. 
$$
Similarly, we write $H_{\alpha,\gamma}^{\rm R}$, $X_{\alpha,\gamma}^{{\rm R},i}$,
$Y_{\alpha,\gamma}^{{\rm R},i}$ and $Z_{\alpha,\gamma}^{{\rm R},i}$, where $x_i$ is replaced by $y_i$. 
We denote by $X^{{\rm L},i}$ the matrix whose $(\alpha,\gamma)$ element is given by $X_{\alpha,\gamma}^{{\rm L},i}$. 
Since all the matrix elements are real, the transpose of $X^{{\rm L},i}$ is equal to its adjoint $(X^{{\rm L},i})^\ast$.  
This property holds for the other quantities $Y,Z,H$, and we will use the same notation. 

Since the coefficients $\mathcal{C}_{\alpha,\beta}$ of the ground state $\Psi$ of (\ref{PsiGS}) is real, 
the expectation value of the Hamiltonian with respect to $\Psi$ can be written as 
\begin{eqnarray*}
\tilde{E}_0^{(\Lambda)}(B,\tilde{f})&=&\langle \Psi,H\Psi\rangle\\
&=&\sum_{\alpha,\beta,\gamma}\mathcal{C}_{\alpha,\beta}\mathcal{C}_{\gamma,\beta}H_{\alpha,\gamma}^{\rm L}
+\sum_{\alpha,\beta,\gamma}\mathcal{C}_{\alpha,\beta}\mathcal{C}_{\alpha,\gamma}H_{\beta,\gamma}^{\rm R}\\
&-&\sum_i\sum_{\alpha,\beta,\gamma,\delta}\mathcal{C}_{\alpha,\beta}\mathcal{C}_{\gamma,\delta}
(X_{\alpha,\gamma}^{{\rm L},i}X_{\beta,\delta}^{{\rm R},i}+Y_{\alpha,\gamma}^{{\rm L},i}Y_{\beta,\delta}^{{\rm R},i}
+Z_{\alpha,\gamma}^{{\rm L},i}Z_{\beta,\delta}^{{\rm R},i})\\
&=&{\rm Tr}\; \mathcal{C}\mathcal{C}^\ast H^{\rm L}+{\rm Tr}\;\mathcal{C}^\ast\mathcal{C}H^{\rm R}\\
&-&\sum_i {\rm Tr}\Bigl[\mathcal{C}^\ast X^{{\rm L},i}\mathcal{C}(X^{{\rm R},i})^\ast
+\mathcal{C}^\ast Y^{{\rm L},i}\mathcal{C}(Y^{{\rm R},i})^\ast
+\mathcal{C}^\ast Z^{{\rm L},i}\mathcal{C}(Z^{{\rm R},i})^\ast\Bigr], 
\end{eqnarray*}
where $\mathcal{C}$ is the matrix whose $(\alpha,\beta)$ element is given by $\mathcal{C}_{\alpha,\beta}$. 
In order to estimate the right-hand side, the following lemma is useful \cite{KLS}: 

\begin{lemma}
Let $\hat{\mathcal{C}}, M,N$ be matrices. Then, the following bound is valid: 
\begin{equation}
\label{traceformula}
|{\rm Tr}\;\hat{\mathcal{C}}^\ast M\hat{\mathcal{C}}N^\ast|^2\le 
\left[{\rm Tr}\; \hat{\mathcal{C}}_{\rm L}M^\ast \hat{\mathcal{C}}_{\rm L}M\right]
\left[{\rm Tr}\; \hat{\mathcal{C}}_{\rm R}N \hat{\mathcal{C}}_{\rm R}N^\ast\right],
\end{equation}
where $\hat{\mathcal{C}}_{\rm L}:=(\hat{\mathcal{C}}\hat{\mathcal{C}}^\ast)^{1/2}$ 
and $\hat{\mathcal{C}}_{\rm R}:=(\hat{\mathcal{C}}^\ast\hat{\mathcal{C}})^{1/2}$. 
\end{lemma}

\begin{proof}{Proof}
By using the polar decomposition, the matrix $\hat{\mathcal{C}}$ can be written as 
$$
\hat{\mathcal{C}}=U\hat{\mathcal{C}}_{\rm R}
$$
in terms of the unitary matrix $U$ and $\hat{\mathcal{C}}_{\rm R}$. 
Clearly, one has $\hat{\mathcal{C}}^\ast=\mathcal{C}_{\rm R}U^\ast$. 
By using this and the cyclicity of the trace, one has 
$$
{\rm Tr}\; \hat{\mathcal{C}}^\ast M\hat{\mathcal{C}}N^\ast
={\rm Tr}\; JK,
$$
where 
$$
J:=\hat{\mathcal{C}}_{\rm R}^{1/2}U^\ast M U\hat{\mathcal{C}}_{\rm R}^{1/2}
\quad \mbox{and}\quad 
K:=\hat{\mathcal{C}}_{\rm R}^{1/2}N^\ast\hat{\mathcal{C}}_{\rm R}^{1/2}. 
$$
The Schwarz inequality for traces yields 
\begin{eqnarray*}
|{\rm Tr}\;\hat{\mathcal{C}}^\ast M\hat{\mathcal{C}}N^\ast|^2&\le&
\Bigl[{\rm Tr}\; J^\ast J\Bigr]\Bigl[{\rm Tr}\; K^\ast K\Bigr]\\
&=&\Bigl[{\rm Tr}\; U\hat{\mathcal{C}}_{\rm R}U^\ast M^\ast U\hat{\mathcal{C}}_{\rm R}U^\ast M\Bigr]
\Bigl[{\rm Tr}\; \hat{\mathcal{C}}_{\rm R}N\hat{\mathcal{C}}_{\rm R}N^\ast\Bigr].
\end{eqnarray*}
For the matrix $U\hat{\mathcal{C}}_{\rm R}U^\ast$ in the right-hand side, 
one has $(U\hat{\mathcal{C}}_{\rm R}U^\ast)^2=U\hat{\mathcal{C}}_{\rm R}^2U^\ast
=\hat{\mathcal{C}}\hat{\mathcal{C}}^\ast$. 
This implies $U\hat{\mathcal{C}}_{\rm R}U^\ast=(\hat{\mathcal{C}}\hat{\mathcal{C}}^\ast)^{1/2}=\hat{\mathcal{C}}_{\rm L}$. 
Substituting this into the above right-hand side, the desired result (\ref{traceformula}) is obtained.   
\end{proof}

In order to apply the inequality (\ref{traceformula}) to the present case, 
we set $\hat{\mathcal{C}}=\mathcal{C}$, $M=X^{{\rm L},i}$ and $N=X^{{\rm R},i}$. As a consequence, one has 
\begin{eqnarray*}
\left|{\rm Tr}\;\mathcal{C}^\ast X^{{\rm L},i}\mathcal{C}(X^{{\rm R},i})^\ast \right|
&\le& \Bigl[{\rm Tr}\;\mathcal{C}_{\rm L}X^{{\rm L},i}\mathcal{C}_{\rm L}(X^{{\rm L},i})^\ast\Bigr]^{1/2}
\Bigl[{\rm Tr}\;\mathcal{C}_{\rm R}X^{{\rm R},i}\mathcal{C}_{\rm R}(X^{{\rm R},i})^\ast\Bigr]^{1/2}\\
&\le&\frac{1}{2}{\rm Tr}\;\mathcal{C}_{\rm L}X^{{\rm L},i}\mathcal{C}_{\rm L}(X^{{\rm L},i})^\ast
+\frac{1}{2}{\rm Tr}\;\mathcal{C}_{\rm R}X^{{\rm R},i}\mathcal{C}_{\rm R}(X^{{\rm R},i})^\ast,
\end{eqnarray*}
where $\mathcal{C}_{\rm L}:=(\mathcal{C}\mathcal{C}^\ast)^{1/2}$ 
and $\mathcal{C}_{\rm R}:=(\mathcal{C}^\ast\mathcal{C})^{1/2}$, 
and we have used $2ab\le a^2+b^2$ for $a,b\in \re$. 
Clearly, a similar inequality holds for matrices $Y$ and $Z$. 
Consequently, the energy expectation value for $\tilde{f}=\overline{f}$ is estimated from below as  
\begin{eqnarray}
\label{E0lowerbound}
\tilde{E}_0^{(\Lambda)}(B,\overline{f})&\ge&{\rm Tr}\; \mathcal{C}_{\rm L}^2 H^{\rm L}
+{\rm Tr}\; \mathcal{C}_{\rm R}^2H^{\rm R}\nonumber\\
&-&\frac{1}{2}\sum_i {\rm Tr}\;\Bigl[\mathcal{C}_{\rm L}X^{{\rm L},i}\mathcal{C}_{\rm L}(X^{{\rm L},i})^\ast
+\mathcal{C}_{\rm L}Y^{{\rm L},i}\mathcal{C}_{\rm L}(Y^{{\rm L},i})^\ast
+\mathcal{C}_{\rm L}Z^{{\rm L},i}\mathcal{C}_{\rm L}(Z^{{\rm L},i})^\ast\Bigr]\nonumber\\
&-&\frac{1}{2}\sum_i {\rm Tr}\;\Bigl[\mathcal{C}_{\rm R}X^{{\rm R},i}\mathcal{C}_{\rm R}(X^{{\rm R},i})^\ast
+\mathcal{C}_{\rm R}Y^{{\rm R},i}\mathcal{C}_{\rm R}(Y^{{\rm R},i})^\ast
+\mathcal{C}_{\rm R}Z^{{\rm R},i}\mathcal{C}_{\rm R}(Z^{{\rm R},i})^\ast\Bigr].\nonumber\\
\end{eqnarray}

Let $f_x^{\rm R}$ be the function such that the value of $f_x^{\rm R}$ 
is equal to $\overline{f}_x$ for the right site $x\in\Lambda^{\rm R}$ 
and that the value of $f_x^{\rm R}$ for the left site $x\in\Lambda^{\rm L}$ 
is equal to the reflection of $\overline{f}_x$ with respect to the planes $\Pi$.
Similarly, the function $f_x^{\rm L}$ is defined. 
Since $\overline{f}_x\ne \overline{f}_y$ for at least one crossing bond, at least one choice, $f^{\rm R}$ or $f^{\rm L}$,  
has the property that it has strictly fewer bonds with $\tilde{f}_x\ne \tilde{f}_y$ than 
does the original function $\overline{f}$.  

Define 
$$
\Psi^{\rm L}:=\sum_{\alpha,\beta}(\mathcal{C}_{\rm L})_{\alpha,\beta}\Psi_\alpha^{\rm L}\otimes\Psi_\beta^{\rm R}
$$
and 
$$
\Psi^{\rm R}:=\sum_{\alpha,\beta}(\mathcal{C}_{\rm R})_{\alpha,\beta}\Psi_\alpha^{\rm L}\otimes\Psi_\beta^{\rm R}
$$
in terms of the matrices $\mathcal{C}_{\rm L}$ and $\mathcal{C}_{\rm R}$. From the definitions, one can easily show 
$\Vert\Psi^{\rm L}\Vert=\Vert\Psi^{\rm R}\Vert=\Vert\Psi\Vert$. 
Then, the right-hand side of (\ref{E0lowerbound}) can be written in terms of the energy expectation values with respect to 
$\Psi^{\rm L}$ and $\Psi^{\rm R}$. Namely, one has 
\begin{eqnarray*}
\tilde{E}_0^{(\Lambda)}(B,\overline{f})&\ge& 
\frac{1}{2}\langle \Psi^{\rm L},\tilde{H}_{\rm p}^{(\Lambda)}(B,f^{\rm L})\Psi^{\rm L}\rangle
+\frac{1}{2}\langle \Psi^{\rm R},\tilde{H}_{\rm p}^{(\Lambda)}(B,f^{\rm R})\Psi^{\rm R}\rangle \\
&\ge& \frac{1}{2}\tilde{E}_0^{(\Lambda)}(B,f^{\rm L})+\frac{1}{2}\tilde{E}_0^{(\Lambda)}(B,f^{\rm R}). 
\end{eqnarray*}
Recall that the function $\overline{f}$ has been chosen so that the energy $\tilde{E}_0^{(\Lambda)}(B,\overline{f})$ 
is a minimum. Therefore, the above inequality implies that 
both of $\tilde{E}_0^{(\Lambda)}(B,f^{\rm L})$ and $\tilde{E}_0^{(\Lambda)}(B,f^{\rm R})$ must take 
the same minimum value. This contradicts the minimality of the number of bonds such that $\overline{f}_x\ne\overline{f}_y$.   
Thus, the inequality (\ref{energytildefbound}) is proved.

%%%%%%%%%%%%%%%%%%%%%%%%%%%%%%%
\Section{Proof of the Bound (\ref{chibound})}
\label{appen:chibound}

As mentioned at the beginning of Appendix~\ref{RP}, we do not assume the uniqueness of the ground state. 

In order to prove the bound (\ref{chibound}), we consider  
$$
H_{\rm p}^{(\Lambda)}(B,\lambda f)=H_{\rm p}^{(\Lambda)}(B)+\lambda H_1'+\frac{\lambda^2}{2}H_2',
$$
where $\lambda$ is a small real parameter, 
$$
H_1'=\sum_{\{x,y\}\subset\Lambda:|x-y|=1}\left[S_x^{(1)}+S_y^{(1)}\right](f_x+f_y)
$$
and 
$$
H_2':=\sum_{\{x,y\}\subset\Lambda:|x-y|=1}(f_x+f_y)^2.
$$
{From} the inequality (\ref{E0bound}), one has 
\begin{equation}
\label{E0bound2}
\frac{\lambda^2}{2}H_2'
+\frac{1}{2\pi iq^{(\Lambda)}(B)}\oint dz\; {\rm Tr}\;[H_{\rm p}^{(\Lambda)}(B)+\lambda H_1']
\frac{1}{z-H_{\rm p}^{(\Lambda)}(B)-\lambda H_1'}\ge E_0^{(\Lambda)}(B) 
\end{equation}
for a sufficiently small $|\lambda|$. 
The second term is calculated as \cite{Kato} 
\begin{eqnarray}
& &\frac{1}{2\pi iq^{(\Lambda)}(B)}\oint dz\; {\rm Tr}\;[H_{\rm p}^{(\Lambda)}(B)+\lambda H_1']
\frac{1}{z-H_{\rm p}^{(\Lambda)}(B)-\lambda H_1'}\nonumber \\
&=&\frac{1}{2\pi iq^{(\Lambda)}(B)}\oint dz\; {\rm Tr}\;[H_{\rm p}^{(\Lambda)}(B)+\lambda H_1']
\left[\frac{1}{z-H_{\rm p}^{(\Lambda)}(B)}\right.\nonumber\\
&+&\frac{1}{z-H_{\rm p}^{(\Lambda)}(B)}\lambda H_1'\frac{1}{z-H_{\rm p}^{(\Lambda)}(B)} 
+\left. \frac{1}{z-H_{\rm p}^{(\Lambda)}(B)}\lambda H_1'\frac{1}{z-H_{\rm p}^{(\Lambda)}(B)}\lambda H_1'
\frac{1}{z-H_{\rm p}^{(\Lambda)}(B)}\right]\nonumber\\
&+&{\cal O}(\lambda^3)\nonumber \\
&=& E_0^{(\Lambda)}(B)+\lambda E_{0,1}^{(\Lambda)}(B)+\lambda^2 E_{0,2}^{(\Lambda)}(B)+{\cal O}(\lambda^3), 
\end{eqnarray}
where 
\begin{eqnarray}
E_{0,1}^{(\Lambda)}(B)&:=&\frac{1}{2\pi iq^{(\Lambda)}(B)}\oint dz\; {\rm Tr}\;
\left[H_1'\frac{1}{z-H_{\rm p}^{(\Lambda)}(B)}\right.\nonumber\\
&+&\left.H_{\rm p}^{(\Lambda)}(B)\frac{1}{z-H_{\rm p}^{(\Lambda)}(B)}H_1'
\frac{1}{z-H_{\rm p}^{(\Lambda)}(B)}\right]\nonumber \\
&=&\frac{1}{2\pi iq^{(\Lambda)}(B)}\oint dz\; z\; {\rm Tr}\;\frac{1}{z-H_{\rm p}^{(\Lambda)}(B)}H_1'
\frac{1}{z-H_{\rm p}^{(\Lambda)}(B)}
\end{eqnarray}
and 
\begin{eqnarray}
E_{0,2}^{(\Lambda)}(B)&:=&\frac{1}{2\pi iq^{(\Lambda)}(B)}\oint dz\; {\rm Tr}\;
\left[H_{\rm p}^{(\Lambda)}(B)\frac{1}{z-H_{\rm p}^{(\Lambda)}(B)}H_1'\frac{1}{z-H_{\rm p}^{(\Lambda)}(B)}H_1'
\frac{1}{z-H_{\rm p}^{(\Lambda)}(B)}\right.\nonumber \\
&+&\left. H_1'\frac{1}{z-H_{\rm p}^{(\Lambda)}(B)}H_1'\frac{1}{z-H_{\rm p}^{(\Lambda)}(B)}\right]\nonumber \\
&=&\frac{1}{2\pi iq^{(\Lambda)}(B)}\oint dz\; z\; {\rm Tr}\;\frac{1}{z-H_{\rm p}^{(\Lambda)}(B)}H_1'\frac{1}{z-H_{\rm p}^{(\Lambda)}(B)}H_1'
\frac{1}{z-H_{\rm p}^{(\Lambda)}(B)}.\nonumber 
\end{eqnarray}
One can easily show that 
$$
E_{0,1}^{(\Lambda)}(B)=\omega_B^{(\Lambda)}(H_1')
$$
and  
$$
E_{0,2}^{(\Lambda)}(B)=\omega_B^{(\Lambda)}(H_1'[1-P_0^{(\Lambda)}(B)][E_0^{(\Lambda)}(B)-H_{\rm p}^{(\Lambda)}(B)]^{-1}H_1').
$$
Substituting these into the inequality (\ref{E0bound2}), one obtains 
\begin{eqnarray}
& &\lambda \omega_B^{(\Lambda)}(H_1')+ \frac{\lambda^2}{2}\sum_{\{x,y\}\subset\Lambda:|x-y|=1}(f_x+f_y)^2 \nonumber\\
& &-\lambda^2\omega_B^{(\Lambda)}(H_1'[1-P_0^{(\Lambda)}(B)][H_{\rm p}^{(\Lambda)}(B)-E_0^{(\Lambda)}(B)]^{-1}H_1')\ge 0.
\end{eqnarray}
This implies 
$$
\omega_B^{(\Lambda)}(H_1')=0
$$
and 
$$
\frac{1}{2}\sum_{\{x,y\}\subset\Lambda:|x-y|=1}(f_x+f_y)^2
-\omega_B^{(\Lambda)}(H_1'[1-P_0^{(\Lambda)}(B)][H_{\rm p}^{(\Lambda)}(B)-E_0^{(\Lambda)}(B)]^{-1}H_1')\ge 0.
$$
The latter is nothing but the desired bound (\ref{chibound}). 

%%%%%%%%%%%%%%%%%%%%%%%%%%%%%%%%%%%%%%%%%%%%%%%%%%%%%%%%%%%%%%
\Section{Proof of Lemma~\ref{lemma:LROtranscorr}}
\label{Appendix:prooflem:LROtranscorr}

Following the method in Sec.~5 in \cite{DLS}, we will give the proof. 
The method in \cite{DLS} is slightly different from that in \cite{FSS} 
although the basic idea by using Bochner's theorem and Fourier transformation is the same.

To begin with, we introduce the Fourier transform of the spin operators as 
\begin{equation}
\hat{S}_p^{(i)}:=|\Lambda|^{-1/2}\sum_{x\in\Lambda}e^{-ipx}S_x^{(i)}, \quad i=1,2,3,
\end{equation}
with the wavevector $p$. By using the translational invariance of the present system, one has 
\begin{equation}
\label{corrtrans}
\langle S_x^{(2)}S_y^{(2)}\rangle_{B,\beta}^{(\Lambda)}
=\frac{1}{|\Lambda|}\sum_p e^{ip(x-y)}\langle \hat{S}_p^{(2)}\hat{S}_{-p}^{(2)}\rangle_{B,\beta}^{(\Lambda)}.
\end{equation}
In order to estimate this right-hand side, we introduce three quantities as  
\begin{equation}
\label{gp}
g_p^{(\Lambda)}(B,\beta):=\frac{1}{2}\left[\langle \hat{S}_p^{(2)}\hat{S}_{-p}^{(2)}\rangle_{B,\beta}^{(\Lambda)}+
\langle \hat{S}_{-p}^{(2)}\hat{S}_{p}^{(2)}\rangle_{B,\beta}^{(\Lambda)}\right],
\end{equation}
$$
b_p^{(\Lambda)}(B,\beta):=\frac{1}{Z_{B,\beta}^{(\Lambda)}}\int_0^1 ds\; 
{\rm Tr}\left[\hat{S}_{-p}^{(2)}\; e^{-s \beta H_{\rm p}^{(\Lambda)}(B)}\; \hat{S}_p^{(2)}\; 
e^{-(1-s)\beta H_{\rm p}^{(\Lambda)}(B)}\right]
$$
and 
\begin{equation}
\label{cp}
c_p^{(\Lambda)}(B,\beta):=\langle[\hat{S}_{-p}^{(2)},[H_{\rm p}^{(\Lambda)}(B),\hat{S}_p^{(2)}]]
\rangle_{B,\beta}^{(\Lambda)}.
\end{equation}
Although the Hamiltonian $H_{\rm p}^{(\Lambda)}(B)$ includes the term of the staggered magnetic field,  
the method of the reflection positivity \cite{DLS} is applicable to the present system. 
As a result, the function $b_p^{(\Lambda)}(B,\beta)$ satisfies the same bound \cite{DLS} 
as in the case of the zero magnetic field, i.e., one has  
\begin{equation}
\label{bpbound}
b_p^{(\Lambda)}(B,\beta)\le (2\beta \mathcal{E}_p')^{-1},
\end{equation}
where 
$$
\mathcal{E}_p':=d+\sum_{i=1}^d \cos p^{(i)}.
$$
In addition to this, the function $b_p^{(\Lambda)}(B,\beta)$ satisfies \cite{FB} 
\begin{equation}
\label{bpboundgpcp}
b_p^{(\Lambda)}(B,\beta)\ge \frac{4[g_p^{(\Lambda)}(B,\beta)]^2}{4g_p^{(\Lambda)}(B,\beta)+\beta c_p^{(\Lambda)}(B,\beta)}, 
\end{equation}
where we have used the inequalities (34) and (A10) in \cite{DLS}, and 
$$
t^{-1}(1-e^{-t})\ge (1+t)^{-1} \quad \mbox{for \ } t>0.
$$
Using the inequality (\ref{bpboundgpcp}), the function $g_p^{(\Lambda)}(B,\beta)$ is estimated as  
\begin{equation}
\label{gpbound}
g_p^{(\Lambda)}(B,\beta)\le 
\frac{1}{2}\left\{b_p^{(\Lambda)}(B,\beta)+\sqrt{[b_p^{(\Lambda)}(B,\beta)]^2
+\beta b_p^{(\Lambda)}(B,\beta) c_p^{(\Lambda)}(B,\beta)}\right\}. 
\end{equation}
The function $c_p^{(\Lambda)}(B,\beta)$ satisfies \cite{DLS}  
\begin{equation}
\label{cpbound}
c_p^{(\Lambda)}(B,\beta)\le 4 S^2\mathcal{E}_p+{\rm Const.}|B|,
\end{equation}
where $S$ is the magnitude of spin, and  
$$
\mathcal{E}_p:=d-\sum_{i=1}^d \cos p^{(i)}.
$$
Combining this, (\ref{bpbound}) and (\ref{gpbound}), one obtains 
$$
g_p^{(\Lambda)}(B,\beta)\le (2\beta \mathcal{E}_p')^{-1}+\sqrt{\frac{S^2 \mathcal{E}_p+{\rm Const.}|B|}{2\mathcal{E}_p'}}. 
$$
{From} the definition (\ref{gp}) of $g_p^{(\Lambda)}(B,\beta)$, we have 
\begin{equation}
\label{omegaCorrTrans}
2g_p^{(\Lambda)}(B,\beta)=\langle \hat{S}_p^{(2)}\hat{S}_{-p}^{(2)}\rangle_{B,\beta}^{(\Lambda)}+
\langle \hat{S}_{-p}^{(2)}\hat{S}_{p}^{(2)}\rangle_{B,\beta}^{(\Lambda)}\ge 
\langle \hat{S}_p^{(2)}\hat{S}_{-p}^{(2)}\rangle_{B,\beta}^{(\Lambda)}, 
\end{equation}
where we have used $(\hat{S}_{-p}^{(2)})^\ast=\hat{S}_{p}^{(2)}$. 
Combining these two inequalities, one has  
\begin{equation}
\label{SpCorrbound}
\langle \hat{S}_p^{(2)}\hat{S}_{-p}^{(2)}\rangle_{B,\beta}^{(\Lambda)}\le (\beta \mathcal{E}_p')^{-1}
+\sqrt{\frac{2S^2 \mathcal{E}_p+{\rm Const.}|B|}{\mathcal{E}_p'}}. 
\end{equation}
This right-hand side is integrable with respect to the wavevector $p$ in three or 
higher dimensions except for the singularity at $p=(\pi,\ldots,\pi)$. 

We write 
$$
F(x-y)=\rho_0(S_x^{(2)}S_y^{(2)}),
$$
where the state $\rho_0$ is given by (\ref{def:rho0}) with (\ref{rhoB}). 
Let $f_x$ be a complex-valued function on $\ze^d$ with a compact support. 
Then, one has 
$$
\sum_{x,y}\overline{f_x}\langle S_x^{(2)}S_y^{(2)}\rangle_{B,\beta}^{(\Lambda)}f_y
=\langle \sum_x \overline{f_x}S_x^{(2)}\sum_y f_yS_y^{(2)}\rangle_{B,\beta}^{(\Lambda)}\ge 0. 
$$
Therefore, in the double limit, $B\searrow 0$ and $\Lambda\nearrow\ze^d$, one can define 
the inner product, 
$$
(f,g)=\sum_{x,y}\overline{f_x}F(x-y)g_y, 
$$
for two functions, $f_x$ and $g_y$, with a compact support. By using Bochner's theorem, the function $F(x-y)$ 
having this property can be written  
\begin{equation}
\label{CorrS2S2inf}
\rho_0(S_x^{(2)}S_y^{(2)})=F(x-y)=\int e^{ip(x-y)}dG_p
\end{equation}
in terms of a measure $G_p$ on the momentum space $p$. (See, for example, Theorem~IX.9 in the book \cite{RS}.)
Clearly, one has  
\begin{equation}
\label{ffrhoCorr}
\sum_{x,y}\overline{f_x}\rho_0(S_x^{(2)}S_y^{(2)})f_y
=\int |\hat{f}(p)|^2dG_p
\end{equation}
for a function $f_x$ with a compact support, where 
$$
\hat{f}(p):=\sum_x e^{-ipx}f_x.
$$

On the other hand, from (\ref{corrtrans}) and (\ref{SpCorrbound}), one has 
\begin{eqnarray*}
\lim_{B\searrow 0}\lim_{\Lambda\nearrow\ze^d}\sum_{x,y}\overline{f_x}f_y
\langle S_x^{(2)}S_y^{(2)}\rangle_{B,\beta}^{(\Lambda)}
&=&\lim_{B\searrow 0}\lim_{\Lambda\nearrow\ze^d}
\frac{1}{|\Lambda|}\sum_p |\hat{f}(p)|^2\langle \hat{S}_p^{(2)}\hat{S}_{-p}^{(2)}\rangle_{B,\beta}^{(\Lambda)}\\
&\le& \frac{1}{(2\pi)^d}\int dp^{(1)}\cdots dp^{(d)} |\hat{f}(p)|^2\left[(\beta \mathcal{E}_p')^{-1}
+\sqrt{\frac{2S^2 \mathcal{E}_p}{\mathcal{E}_p'}}\right]
\end{eqnarray*}
whenever $\hat{f}(p)=0$ for $p=(\pi,\ldots,\pi)$. Here, the double limit is the same as the weak$^\ast$-limit 
for the state $\rho_0$. 
Since this left-hand side is equal to the left-hand side of (\ref{ffrhoCorr}) from the definition of the state $\rho_0$, 
one notices that \cite{FSS,DLS} the measure $G_p$ consists of a delta measure at $p=(\pi,\ldots,\pi)$ and 
a absolutely continuous part in $p$. 
Therefore, the application of Riemann-Lebesgue theorem to the right-hand side of 
the correlation function of (\ref{CorrS2S2inf}) yields (\ref{LROtranscorr}). 
Namely, the contribution of the absolutely continuous part of the measure $G_p$ is vanishing 
in the limit $|x-y|\rightarrow\infty$.  
The nonvanishing contribution may come from the delta measure at the singularity.

%%%%%%%%%%%%%%%%%%%%%%%%%%%%%%%%%%%%%%%%%%%%%%%%%%%%%%%%%%%%%%%%%%%%%%%%%%%
\bigskip\bigskip\bigskip

\noindent
{\bf Acknowledgements:} I would like to thank Akinori Tanaka and Hal Tasaki for helpful comments and discussions. 
%%%%%%%%%%%%%%%%%%%%%%%%%%%%%%%%%%%%%%%%%%%%%%%%%%%%%%%%%%%%%
%\newpage

\end{document}